\documentclass[a4paper,11pt]{article}
\usepackage{geometry}
\usepackage{a4wide}
\usepackage[UKenglish]{babel}

\usepackage[noadjust]{cite}
\usepackage{amsmath}
\usepackage{amsfonts}
\usepackage{amsthm}
\usepackage{amssymb}
\usepackage{mathtools}
\usepackage{microtype}
\usepackage{bm}
\usepackage{thm-restate}
\usepackage{hyperref}
\hypersetup{colorlinks=true,citecolor=blue,linkcolor=blue,urlcolor=blue}
\usepackage{cleveref}
\usepackage{complexity}
\usepackage{mathrsfs}
\usepackage{scalerel}
\usepackage{enumitem}

\let\orignot\not %
\usepackage{mathabx} %
\let\not\orignot

\newcommand{\Cos}[1]{\mathrm{Cos}\left(#1\right)}
\DeclareMathOperator{\pcsp}{PCSP}
\DeclareMathOperator{\csp}{CSP}
\DeclareMathOperator{\pol}{Pol}
\DeclareMathOperator{\PMC}{PMC}

\DeclareMathOperator{\im}{Im}

\newcommand{\Yes}{Yes} %
\newcommand{\No}{No} %

\newcommand\ignore[1]{}

\newcommand\rel[1]{\mathbf{#1}}
\newcommand\alg[1]{{#1}}

\newcommand{\I}{{\bm{I}}}

\newcommand{\Hcal}{\mathcal{H}}

\newcommand{\Mscr}{\mathscr{M}}
\newcommand{\NN}{\mathbb{N}}
\newcommand{\ZZ}{\mathbb{Z}}
\newcommand{\arty}{\mathrm{ar}}
\newcommand{\little}[1]{\scaleto{#1}{5pt}}

\newcommand\Crestrict[2]{{%
  \left.\kern-\nulldelimiterspace %
  #1 %
  \right|_{#2} %
  }}

\newcommand\abreg[1]{{#1}^{\textnormal{r.c.}}}

\newtheorem{theorem}{Theorem}[section]
\newtheorem{lemma}[theorem]{Lemma}
\newtheorem{corollary}[theorem]{Corollary}
\newtheorem{observation}[theorem]{Observation}

\newtheorem{proposition}[theorem]{Proposition}

\theoremstyle{definition}

\begin{document}

\title{Equations over Finite Monoids with Infinite Promises\thanks{This work was supported by UKRI EP/X024431/1. Work done while Alberto Larrauri was at the University of Oxford. For the purpose of Open Access, the authors have applied a CC BY public copyright licence to any Author Accepted Manuscript version arising from this submission. All data is provided in full in the results section of this paper.}}

\author{Alberto Larrauri\\University of Zaragoza \and Antoine Mottet\\Hamburg University of Technology \and Stanislav
\v{Z}ivn\'y\\University of Oxford}

\date{}

\maketitle
\begin{abstract}
  Larrauri and \v{Z}ivn\'y~[ICALP'24/ACM ToCL'24] recently established a complete complexity
  classification of the problem of solving a system of equations over a 
  monoid $\alg N$ assuming that a solution exists over a monoid
  $\alg M$, where both monoids are finite and $\alg M$ admits a homomorphism to
  $\alg N$. Using the algebraic approach to promise constraint satisfaction
  problems, we extend their complexity classification in two directions: we obtain a complexity dichotomy in the case where arbitrary relations are added to the monoids, and we moreover allow the monoid $\alg M$ to be finitely generated.
\end{abstract}

\section{Introduction}
\paragraph*{Solving Equations}
Deciding the solvability of systems of equations is a fundamental  problem in computer science and mathematics.
In general, the equation satisfiability problem for an algebraic structure $\alg A$ takes as input a set of equations of the form
\[ s(x_1,\dots,x_n) = t(x_1,\dots,x_n)\]
where $x_1,\dots,x_n$ are variables and $s,t$ are terms over the algebraic symbols of the structure.
For example, over a monoid the terms would be simply words over the alphabet $\{x_1,\dots,x_n\}$, over groups the terms would be words over the alphabet $\{x_1,x_1^{-1},\dots,x_n,x_n^{-1}\}$, and over rings the terms would be multivariate polynomials.
For now, the reader can imagine that the terms may also contain constants from the algebraic structure $\alg A$, but we will revisit this assumption in the following.
The question is to decide whether there exists an assignment $h\colon\{x_1,\dots,x_n\}\to \alg A$ such that every equation becomes true under this assignment.

Besides their intrinsic mathematical appeal, equation satisfiability problems
are  very natural for computer science, in particular in the case where the equations are to be solved over a finitely generated free structure.
Indeed, the elements of such structures are themselves words or terms over the generators and the equation satisfiability problem then coincides with the unification problem, 
where on an input consisting of equations as above, the goal is to determine the existence of a substitution of the variables into terms that would make the equations true, possibly modulo a background equational theory.
For example, unification problems over words with a finite alphabet $\Sigma$ correspond to solving equations in the free monoid generated by $\Sigma$.
Unification is relevant to areas such as automated reasoning and description
logics~\cite{HandbookAutomatedReasoningBaader,UnificationACUI}, among many others.

Depending on the underlying algebraic structure over which the equations are to be solved, and depending on the type of equations to be solved, the complexity of the problem varies wildly.
While systems of linear equations over a field, where each term $s,t$ in the input is a linear polynomial, are known to be efficiently solvable to any undergraduate student, solving arbitrary polynomial equations is NP-hard and is not considered to be efficiently solvable.
It is also known that the equation satisfiability problem in the group $(\ZZ,+)$
(or any finitely generated free \emph{commutative} group) is solvable in polynomial
time~\cite{Kannan}, while the problem becomes much more complex over arbitrary
finitely generated free groups or monoids, for which only algorithms requiring
polynomial space are known~\cite{Plandowski04:jacm}.

Over \emph{finite} structures, it has been known for several decades that the equation satisfiability problem for finite groups is solvable in polynomial time when the group is commutative, and is NP-complete otherwise~\cite{Goldmann02:ic}.
For finite monoids, a similar dichotomy has been obtained by Kl\'ima, Tesson, and Th\'erien: the equation satisfiability problem is solvable in polynomial time for \emph{{regular}} commutative monoids, and is NP-complete otherwise \cite{KTT07:tcs}.

\paragraph*{Generalized Equations}
From here on we do not necessarily allow constants to appear in the instances of the equation satisfiability problem. We outline two possible approaches to obtaining polynomial-time algorithms solving more general problems than equation satisfiability. 

On the one hand, one can start with a tractable equation satisfiability problem and additionally allow in the input constraints that are \emph{not} equations.
A natural example for this is to allow e.g.\ constraints of the form $s(x_1,\dots,x_n)\neq t(x_1,\dots,x_n)$ or $(x_1,\dots,x_r)\in R$, where $R\subseteq\alg A^r$ is a fixed subset.
The appropriate framework to study these problems is that of \emph{constraint satisfaction problems} whose templates are expansions $\rel A=(\alg A, R)$ of algebraic structures by a relation; we introduce this terminology formally in~\Cref{sect:prelims}.
For finite commutative groups, Feder and Vardi~\cite{Feder98:monotone} showed that allowing any constraint that is not itself expressible by a system of equations (so-called cosets) yields an NP-hard problem.
Interestingly, this is not true for finitely generated commutative groups~\cite{CSPExpansionsFreeAbelianGroups} and in particular it is possible in this setting to decide the satisfiability of a system of equations together with disequality constraints as above.

On the other hand, one can start with an NP-hard equation satisfiability problem, and \emph{restrict} the instances allowed as input.
This has been investigated e.g. in~\cite{LZ24:acm} in the following form. Fix two finite monoids $\alg M,\alg N$. The \emph{promise} equation problem only allows systems of equations that are promised to be satisfiable in the monoid $\alg M$ or unsatisfiable in $\alg N$, and the problem is to decide which case applies. When $\alg M=\alg N$, we see that we recover exactly the classical equation satisfiability problem. Note that to specify this problem formally, one needs to also specify a translation between the constant symbols corresponding to elements of $\alg M$ and those corresponding to elements of $\alg N$; this is more than a technical detail, as the complexity of the problem also depends on this translation.
The suitable framework to study such problems is that of \emph{promise}
constraint satisfaction problems~\cite{BG21:sicomp,BBKO21} whose templates are monoids $\alg M$ and $\alg N$.

\paragraph*{Main Result} %
In this work, we generalize the state-of-the-art in both aforementioned directions  simultaneously.
Let $\alg M, \alg N$ be monoids, and let $R(M)\subseteq\alg M^r$ and $R(N)\subseteq \alg N^r$ be relations of some arity $r$.
The problem that we consider takes as input a system of constraints that are either monoid or group equations of the form $s(x_1,\dots,x_n)=t(x_1,\dots,x_n)$ as considered above, or constraints of the form $(x_1,\dots,x_r)\in R$.
The problem is to decide whether this system is satisfiable in $\alg M$ (with $R$ interpreted as $R(M)$), or not satisfiable in $\alg N$ (with $R$ interpreted as $R(N)$). 
Until proper terminology is introduced below, we call this problem the \emph{generalized equation problem}.

For example, consider $\alg M$ to be the monoid $(\ZZ_{\geq 0},+)$ of non-negative integers, and $\alg N$ to be the monoid $(\mathbb Z/n\mathbb Z,+)$ for an arbitrary $n\geq 2$. 
Let $R(M)$ contain all the triples obtained by permuting entries of tuples of the form $(a,a,a+1)$ for $a\in\mathbb N$. Let $R(N)$ be the relation containing all triples except for $(a,a,a)$ for all $a\in\mathbb Z/n\mathbb Z$.
An example of an instance to this problem is
$\{x+y = u+v, \, (x,y,u)\in R, (u,v,x) \in R, (u,v,y) \in R\}$, which is not satisfiable in $\rel M=(\alg M, R(M))$, but is satisfiable in $\rel N=(\alg N,R(N))$ for all $n\geq 2$.

We note that deciding the satisfiability of such instances in $\rel M$ is
NP-complete.\footnote{NP containment holds, since for example this problem reduces to the existential theory of Presburger arithmetic.}
Indeed, 1-in-3-SAT readily reduces to instances of $\rel M$ that do not even use monoid equations over $\alg M$ and only use the constraints of the form $R(M)$. 
Similarly, deciding satisfiability of such instances in $\rel N$ is
NP-complete for all $n\geq 2$ by a reduction from Not-All-Equal-SAT.
However, when the problem is to distinguish between instances satisfiable in
$\rel M$ and those unsatisfiable in $\rel N$, then the problem becomes solvable
in polynomial time when $n$ is a multiple of $3$ and remains NP-hard
otherwise.\footnote{The regularization of $\rel M$ is $(\mathbb Z, +)$ and any
monoid homomorphism to $\rel N$ contains in its kernel the set $n\mathbb Z$ of multiples of $n$.
The coset $\Cos{R(M)}$ consists of all the triples $(a,b,c)$ such that $a+b+c=1 \bmod 3$. In particular, for all $n$ not divisible by 3, it contains a tuple $(a,b,c)$ such that $a=b=c\bmod n$, and therefore there is no homomorphism from $(\mathbb Z,+,\Cos{R(M)}$
to $(\mathbb Z/n\mathbb Z,+,R(N))$. This implies, by our main result
(cf.~\Cref{thm:main}), that $\pcsp(\rel
M,\rel N)$ is NP-hard. In contrast, if $n$ \emph{is} divisible by 3 then no
tuple $(a,b,c)$ such that $a=b=c \bmod 3$ is in $\Cos{R(M)}$, and therefore the canonical projection $x\mapsto x\bmod 3$ 
 is a homomorphism from $(\mathbb Z,+,\Cos{R(M)})$ to $(\mathbb
Z/3\mathbb Z,+,R(N))$. By~\Cref{thm:main}, in this case $\pcsp(\rel
M,\rel N)$ is solvable in polynomial time.}

We are able to obtain a dichotomy result for the class of such problems where $\alg M$ is finitely generated and where $\alg N$ is a finite monoid.
\begin{theorem}[Main theorem, informal version]\label{main-informal}
	Let $\alg M$ be a finitely generated monoid, and let $\alg N$ be a finite monoid. Let $R(M)\subseteq\alg M^r$ and $R(N)\subseteq\alg N^r$ be arbitrary relations of some arity $r$.
	Then the generalized equation problem parameterized by $\rel M=(\alg M,R(M))$ and $\rel N=(\alg N,R(N))$ is
  either solvable in polynomial time or is NP-hard.
\end{theorem}

We remark that this result generalizes the main dichotomy from \cite{LZ24:acm}.
At a high level, the classification in \cite{LZ24:acm} involves templates
consisting of pairs of finite monoids $\alg M$, and $\alg
N$\footnote{The phrase ``infinite promises'' in the tile of this article
alludes to the fact that $\alg M$ can be infinite (if finitely generated).} enriched with constants $c_1(M),\dots, c_\ell(M)$, 
and $c_1(N),\dots, c_\ell(N)$. Such a problem can be expressed as in the statement
of~\Cref{main-informal} by letting $R(M)=\{(c_1(M), \dots, c_\ell(M))\}$ and
$R(N)=\{(c_1(N),\dots, c_\ell(N))\}$.

\paragraph*{Other Related Work}

In~\cite{ExistentialTheoryFreeGroups}, the authors consider the complexity of the \emph{uniform} CSP over expansions of finitely generated free groups.
Since the constraint relations are here part of the input, a fixed encoding of the constraints needs to be agreed upon, and the authors represent constraints by regular expressions over the generators of the group.
In particular, not all possible constraints are allowed as input.
Under those conditions, the problem is shown to be solvable in polynomial space.

The literature on the equation satisfiability problem over monoids that are
\emph{not} finitely generated is sparse as such problems are almost always
intractable; see e.g.~\cite{SkolemArithmetic} where the problem over $(\mathbb
N;\times)$, the free commutative monoid with countably many generators,
is investigated.

Finally, an important related problem is the problem of trying to maximize the number of satisfied equations, in cases where the input system cannot be fully satisfied.
The promise version of this problem has recently been considered
in~\cite{BLZ25:icalp}, where it is shown that beating the random assignment is
NP-hard, just as in the non-promise setting for systems of equations over the commutative~\cite{Hastad01} and non-commutative
groups~\cite{Engebretsen04:tcs}. A natural and interesting open problem is to
consider the complexity of this problem when constraints other than equations
are also allowed on the input.

\section{Background}\label{sect:prelims}

We adopt the convention that the set of natural numbers $\NN$ begins at $1$. Given an integer $n\geq 1$, we write $[n]$ for the set $\{1, \dots, n\}$. We implicitly extend functions $f\colon U \rightarrow V$ to arbitrary Cartesian powers by coordinate-wise application. In other words, we write $f(u_1,\dots, u_n)$ for the tuple $(f(u_1), \dots , f(u_n))$ for all $(u_1, \dots, u_n)\in U^n$. 
Given a tuple $\bm u =(u_1,\dots, u_m)\in U^m$ and a map $\sigma\colon [n]\to [m]$, we write
$\bm u \circ \sigma$ for the tuple $(u_{\sigma(1)}, \dots, u_{\sigma(n)})\in U^n$.

We assume basic familiarity with the notions of semigroup, group, monoid, and their homomorphisms (see, for instance~\cite{howie1995fundamentals}). Nevertheless, we recall some of the definitions in order to introduce notation. \par

\paragraph*{Semigroups, Monoids, and Groups} 
We introduce here several basic algebraic notions. We warn the reader that our notion of inverse in a monoid deviates from standard references (e.g., \cite{howie1995fundamentals, clifford1961algebraic, grillet2017semigroups}).

A \emph{semigroup} $\alg S$ is a set $S$ equipped with a binary associative operation $a \cdot^{\little{\alg S}} b$. In most cases the semigroup is clear from the context and we write $ab$ for $a\cdot^{\little{\alg S}} b$. A monoid $\alg M$ is a semigroup that additionally contains an element $e_{\little{M}}\in \alg M$ satisfying
$e_{\little{M}} a = a e_{\little{M}} = a$ for all $a\in M$, which we call the \emph{identity element}. A group $\alg G$ is
a monoid where for each element $a\in G$ there is another element $a^{-1}\in G$
satisfying $a a^{-1}= a^{-1} a = e_{\little{G}}$. We call a semigroup $\alg S$
\emph{commutative} if $ab=ba$ for each $a, b\in S$. An element $a\in \alg S$ is
called \emph{idempotent} if $a^2=a$. A \emph{semilattice} is a commutative monoid where each element is idempotent.  
Given a semigroup (resp., monoid, group) $\alg S$ and an integer $n\geq 1$, the Cartesian power $\alg S^n$ inherits the semigroup (resp., monoid, group) structure from $\alg S$ in the natural way. 
\par
A monoid homomorphism $f$ from a monoid $\alg M$ to a monoid $\alg N$, denoted $f\colon \alg M \to \alg N$, is a map $f\colon M \to N$ satisfying both
$f(e_{\little{M}})=e_{\little{N}}$ and 
$f(a b)= f(a)f(b)$ for all $a, b\in M$. We write $\alg M\to \alg N$ to denote the fact that $\alg M$ maps homomorphically to $\alg N$. \par

Let $\alg M$ be a monoid and $U, V\subseteq M$ be sets. We define 
$U \otimes V$ as the set 
$\{
s t \, \vert \, s\in U, t\in V \}$.
Similarly, given $n\geq 1$, we define
$U^{\otimes n}$ as 
$\bigotimes_{i=1}^n U$, whereas $U^n$ denotes the $n$-th Cartesian power of $U$.
Finally, we write $\langle U \rangle$ to denote the submonoid of $\alg M$
generated by $U$, which is the smallest submonoid of $\alg M$ containing all
elements of $U$.
We say that $S$ \emph{generates} $\alg M$ if $\langle S \rangle= \alg M$.

Let $\alg S$ be a semigroup. A subsemigroup $\alg T\leq  \alg S$ 
is called a subgroup of $\alg S$ if there is an element $e_T\in \alg T$ that acts as the identity element in $\alg T$, and each element $s\in \alg T$ has an inverse $t\in \alg T$ satisfying $st=ts=e_T$. Observe that even if $\alg S$ is a monoid, a subgroup $\alg T \leq \alg S$ is not necessarily a submonoid: $T$ does not need to contain the identity $e_{S}$.
A commutative semigroup $\alg S$ is called \emph{{regular}} if every element in $\alg S$ belongs to a subgroup \footnote{
We remark that in the non-commutative setting, there is a difference between regular semigroups $S$, where it is required that for each $a\in S$ there is some $b\in S$ such that 
$aba=a$ and $bab=b$, and completely regular semigroups $S$ (also called unions of subgroups \cite{grillet2017semigroups}), where each element belongs to a subgroup. For commutative semigroups both notions coincide, so there should be no confusion regarding this.}. 
Under P$\neq$NP, systems of equations over a finite monoid $\alg M$ can be
solved in polynomial time if, and only if, $\alg M$ is commutative and {regular}~\cite{KTT07:tcs,LZ24:acm}.
Given a semigroup $S$, the preorder $a \preceq b$
contains all the pairs $a,b\in S$ for which either $a=b$ or there exist elements $c_1,c_2$ satisfying $c_1b =bc_2 = a$. We write $a\Hcal b$ whenever both $a \preceq b$ and $b\preceq a$ hold. It can be seen that $\Hcal$ is an equivalence relation. This way, we denote the $\Hcal$-class of an element $a\in S$ by $H_a$. The following is sometimes referred to as Green's Theorem:

    \begin{theorem}[{\cite[Theorem 2.2.5]{howie1995fundamentals}}]
    \label{th:green_theorem}
    Let $S$ be a semigroup, and $T\subseteq S$ be a $\Hcal$-class. Then $T$ is a
    subgroup of $S$ if, and only if, $T$ contains some idempotent element. In particular, no $\Hcal$-class of $S$ contains more than one idempotent element. \end{theorem}

This result implies that the maximal subgroup containing an idempotent element
$d\in S$ is precisely $H_d$. 
Additionally, an element $a\in S$
belongs to a subgroup if, and only if, $a \Hcal d$ for some idempotent element $d$. In this situation we say that $a$ is a \emph{group element} of $S$
and write $a^{-1}$ for its inverse in $H_a$. Equivalently, $a^{-1}$ is the only element satisfying $aa^{-1} = a^{-1}a$ together with $a^2 a^{-1} =a$ and $a^{-2} a = a^{-1}$, where we use $a^{-k}$ to denote $(a^{-1})^k$.
\footnote{This notion of inverse is less general than the commonly adopted notion of inverse for semigroups, where it is only required that $aba=a$ and $bab=b$ (e.g., \cite{howie1995fundamentals}). In that setting inverses are not necessarily unique.}
Given a set $S\subseteq M$ of group elements we write 
$S^{-1}$ for the set $\{ \, a^{-1} \, \vert \, a\in S \, \}$. 
\par

Let $\alg M$ be a commutative monoid. A \emph{coset} of $\alg M$ is a set $U\subseteq \alg M$ of group elements satisfying
$U \otimes (U^{-1} \otimes U) = U$. Given a set $U\subseteq \alg M$ of group elements, the \emph{coset generated by $U$}, denoted by $\Cos{U}$, is defined as $U\otimes \langle U^{-1} \otimes U \rangle$. It can be seen that $\Cos{U}$ is a coset for all $U\subseteq \alg M$ and it is also the intersection of all cosets containing $U$. We remark that in the particular case where $\alg M$ is a group our notion of coset coincides with the usual one.

\paragraph*{Relational Structures}
A \emph{relational signature} $\Sigma$ is a set of symbols where each $R\in \Sigma$ has a number $\arty(R)$ associated to it, called its arity. In this work we assume relational signatures to always be finite.

Given a relational signature $\Sigma$, a $\Sigma$-structure $\rel A$ consists of a set $A$, called the universe of $\rel A$, and a set $R(\rel A) \subseteq A^{\arty(R)}$ for each symbol $R\in \Sigma$, which is called the interpretation of $R$ on $\rel A$. A \emph{relational structure} is a $\Sigma$-structure for some relational signature $\Sigma$.\par
A homomorphism $f$ from a relational structure $\rel A$ to another relational structure $\rel B$ with the same signature $\Sigma$ is a map
$f:A\to B$ satisfying $f(R(\rel A)) \subseteq R(\rel B)$ for each $R\in \Sigma$. We write $f: \rel A \to \rel B$ to denote that $f$ is a homomorphism from $\rel A$ to $\rel B$, and $\rel A \to \rel B$ to denote that there exists such a homomorphism without specifying one. \par 
Let $\Sigma_\textnormal{mon}$ be the relational signature consisting of a unary symbol $R_e$ and a ternary symbol $R_\otimes$.
Given a monoid $\alg M$, there is a natural $\Sigma_\textnormal{mon}$-structure $\rel M$ associated to it whose universe is $M$,
and where $R_e(\rel M)= \{e_{\little{M}} \}$, and $R_\otimes(\rel M) = \{
(a,b,c) \in M^3 \, \vert \, ab = c \}$.
Note however that not every $\Sigma_\textnormal{mon}$-structure is associated to a monoid.
Observe that given two monoids $\alg M, \alg N$, a map $f\colon M \to N$ is a
homomorphism from $\alg M$ to $\alg N$ if, and only if, it is a homomorphism
from $\rel M$ to $\rel N$. If $\alg M$ is a monoid and $R\subseteq M^n$, we denote by $(\alg M, R)$ the \emph{relational} structure with the three relations $R_e(\rel M)$, $R_\otimes(\rel M)$, and $R$.

\paragraph*{Promise Constraint Satisfaction Problems}
Let $\rel A, \rel B$ be relational structures such that  $\rel A\to\rel B$.
The \emph{promise constraint satisfaction problem} $\pcsp(\rel A,\rel B)$ is the computational problem of distinguishing, for a given input $\rel X$, between the following two cases: (1) there exists a homomorphism $\rel X\to\rel A$, and (2) there does not exist any homomorphism $\rel X\to\rel B$.
Instances falling into the first case are called \Yes-instances, instances
falling into the second case are called \No-instances. We write $\csp(\rel A)$ for the problem $\pcsp(\rel A, \rel A)$.

Note that by our convention, if $\rel M$ and $\rel N$ are expansions of monoids, the input to $\pcsp(\rel M,\rel N)$ is \emph{not} assumed to be an expansion of a monoid, but a relational structure.
This differs, e.g., from~\cite{BartoDemeoMottet,FiniteAlgebras}, where the CSPs of expansions of algebraic structures are investigated and where the inputs are assumed to be themselves algebras.
An input $\rel X$ to $\pcsp(\rel M,\rel N)$ can then be seen as a finite set of equations of the form $s(x_1,\dots,x_n) = t(x_1,\dots,x_n)$ where $s,t$ are words over $\{x_1,\dots,x_n\}$, together with additional constraints involving the extra relations.

\section{Dichotomy and Finite Tractability}
\label{sec:main}

We are now in position to state our main result (\Cref{main-informal}) formally.

\begin{restatable}{theorem}{main}\label{thm:main}
    Let $\rel M=(\alg M,R(M)),\rel N=(\alg N,R(N))$ be expansions of monoids such
    that $\rel M\to\rel N$. Assume that $\alg M$ is finitely generated and $\alg
    N$ is finite. Then, either (1) there exists a homomorphism $h\colon\rel M
    \to \rel N$ whose image is a regular commutative monoid and such that $\Cos{h(R(M))} \subseteq
    R(N)$, in which case $\pcsp(\rel M,\rel N)$ is solvable in polynomial time, or (2) $\pcsp(\rel M,\rel N)$ is NP-hard.
    \end{restatable}

We note that our result also captures templates where $\rel M=(\alg M, R_1(M),\dots, R_\ell(M))$ and $\rel N=(\alg N, R_1(N),\dots, R_\ell(N))$ are expansions of monoids by finitely many non-empty relations. To see this, observe that $\pcsp(\rel M, \rel N)$ is polynomial time equivalent to $\pcsp(\rel M^\prime, \rel N^\prime)$ where
$\rel M^\prime, \rel N^\prime$ are the expansions of $\alg M, \alg N$ by the relations $R(M), R(N)$ obtained by concatenating the tuples in $R_1(M), \dots, R_\ell(M)$ and in
$R_1(N), \dots, R_\ell(N)$ respectively. \par
    
Interestingly, our proof of NP-hardness in~\Cref{thm:main} makes use of the algebraic approach to promise constraint satisfaction~\cite{BBKO21} albeit in a setting where the left-hand side structure is infinite and the right-hand side structure is finite.
In~\cite{BBKO21}, many of the statements for finite structures hold for templates where the \emph{right-hand} side template is infinite and the \emph{left-hand} side template is finite, however the direction that we take here is essentially new and works for arbitrary algebraic structures (i.e., not only groups and monoids) under the natural condition that the structures and their powers are finitely generated.
Such structures with finitely generated left-hand side and finite right-hand
side also appear in the study of (non-promise) constraint satisfaction problems
for infinite structures with a large automorphism group~\cite{Mottet25:lics}.

The polynomial-time algorithm that we use to prove~\Cref{thm:main} is inspired
by the work of~\cite{KTT07:tcs} on finite regular commutative monoids.
We show that the CSP of an expansion of a finitely generated regular commutative monoid by any finite number of cosets is solvable in polynomial time.
Although the monoid $\alg M$ in~\Cref{thm:main} is only assumed to be finitely
generated, we can apply our algorithm by showing that every finitely generated
monoid $\alg M$ admits a \emph{commutative regularization} $(\abreg{\alg
M},\abreg{\pi})$, as given by the following result.
This fact and this proof might be folklore; we include the proof in the appendix (see~\Cref{ap:ab_reg}) for the convenience of the reader.

\begin{restatable}{lemma}{abelian_regularization}
\label{le:abelian_regularization}
    Let $\alg M$ be a monoid.
    Then there exists a {regular} commutative monoid $\abreg{\alg M}$ and a homomorphism $\abreg{\pi}_M\colon \alg M \to \abreg{\alg M}$ satisfying that for any monoid homomorphism $f\colon \alg M \to \alg N$ where $\alg N$ is commutative and {regular}, there is a homomorphism $f^\prime\colon \abreg{\alg M} \to \alg N$
    such that $f= f^\prime \circ \abreg{\pi}_M$. Moreover, if $S\subseteq \alg M$ is a set generating $\alg M$, then $\abreg{\pi}_M(S)$ generates $\abreg{\alg M}$.
\end{restatable}

Let $\abreg{\rel M}=(\abreg{\alg M}, \Cos{\abreg{\pi}(R(M))})$. Then
$\csp(\abreg{\rel M})$ is solvable in polynomial time by our algorithm
briefly discussed above, cf. \Cref{sec:tractability} for all
details. \Cref{thm:main} together with the universal property of $\abreg{\rel M}$ imply that this structure serves as a classifier for the complexity of problems of the form $\pcsp(\rel M,\underline{\hspace{0.2cm}})$.
\begin{corollary}
\label{abreg-sandwich}
Assume P${}\neq{}$NP.
    Suppose that $\rel M=(\alg M,R(M)),\rel N=(\alg N,R(N))$ are expansions of monoids such that $\alg M$ is finitely generated and $\alg N$ is finite.
    Then $\pcsp(\rel M,\rel N)$ is solvable in polynomial time if, and only if,
    there exists a homomorphism $\abreg{\rel M}\to\rel N$.
\end{corollary}

Another implication of \Cref{thm:main} is that a similar dichotomy holds if the templates $\rel M,\rel N$ are expansions of \emph{groups} instead of monoids.
In that case, the terms appearing in an input to $\pcsp(\rel M,\rel N)$ can also use the inverse symbol. However, up to introducing a new variable $x_{-}$ and constraints $x_{-}x=e=xx_{-}$ for each variable $x$, we reduce ``group instances'' to ``monoid instances''.
Moreover, a group is generated by $U$ if, and only if, it is generated as a monoid by $U\cup U^{-1}$.
\begin{corollary}
    Let $\rel G=(\alg G,R^G),\rel H=(\alg H,R^H)$ be expansions of groups such
    that there exists a homomorphism $\rel G\to\rel H$. Assume that $\alg G$ is
    finitely generated and $\alg H$ is finite. Then, either (1) there exists a
    homomorphism $h\colon\rel G \to \rel H$ with a commutative image such that
    $[h(R^G)] \subseteq R^H$, in which case $\pcsp(\rel G,\rel H)$ is solvable in
    polynomial time, or (2) $\pcsp(\rel G,\rel H)$ is NP-hard.
\end{corollary}

The problem $\pcsp(\rel A,\rel B)$ is said to be \emph{finitely tractable} if there exist a finite structure $\rel C$ and homomorphisms $\rel A\to\rel C$ and $\rel C\to\rel B$ such that $\csp(\rel C)$ is solvable in polynomial time.
The study of finite tractability in promise constraint satisfaction is important, as it highlights some essential differences between the non-promise setting (where $\rel A=\rel B$).
It is known that there exist problems $\pcsp(\rel A,\rel B)$ with $\rel A,\rel
B$ finite that are solvable in polynomial time but that are not finitely
tractable~\cite{BBKO21,AB21}. 
A simple consequence of~\Cref{thm:main} is that such problems do not appear in the class of problems under consideration.
\begin{corollary}
Assume P${}\neq{}$NP.
    Suppose that $\rel M,\rel N$ are expansions of monoids $\alg M,\alg N$ such that $\alg M$ is finitely generated and $\alg N$ is finite.
    If $\pcsp(\rel M,\rel N)$ is solvable in polynomial time, then it is finitely tractable.
\end{corollary}

\section{The Structure of Regular Commutative Monoids}
\label{sec:algebra}

In this section we introduce some elementary facts about monoids that will be used in our main results. We refer to \Cref{ap:algebra} for the short proofs, which we include for completeness although they might be known.
We start with the following characterization of {regular} commutative monoids, which is a consequence of Green's theorem. 

\begin{proposition}[{\cite[Proposition 4.1.1]{howie1995fundamentals}}]\label{prop:charRM}
Let $\alg M$ be a commutative monoid. Then $\alg M$ is {regular} if, and only if, each of its $\Hcal$-classes contains an idempotent element.     
\end{proposition}

\par
Given a commutative monoid $\alg M$, we write $M_I$ for its set of idempotent elements. \Cref{prop:charRM} implies that if $\alg M$ is {regular} we can write $\alg M= \bigsqcup_{d\in M_I} \alg H_d$, where
$\alg H_d$ is the $\Hcal$-class of $d\in M_I$, which also is the maximal subgroup containing this element, by \ref{th:green_theorem}. 
The set $M_I$ has the following nice properties.

\begin{restatable}{lemma}{structureabelianmonoid}
    \label{structure-abelian-monoid}
    Let $\alg M$ be a commutative monoid and $M_I\subseteq M$ be the subset of its
    idempotent elements. Then, (1) $\alg M_I$ is a semilattice, and
    (2) if $\alg M$ is {regular} and finitely generated, then $\alg M_I$ is finite.%
\end{restatable}

It is well-known that for any finite group $\alg G$ there is a constant $C\in \NN$ such that $a^C=e_{\little{G}}$ for all $a\in G$. For monoids the following analogue holds.

\begin{restatable}{lemma}{idempotentcst}
\label{le:idempotent_constant}
    Let $\alg M$ be a finite monoid. There is an integer $C>1$ such that $a^C$ is idempotent for all $a\in \alg M$. If $a\in \alg M$ is a group element, then it also holds that $a^{C-1}=a^{-1}$. Additionally, given $a\in \alg M$, there is a unique idempotent element $d_a$ that satisfies $d_a=a^n$ for some $n\in \NN$.
\end{restatable}

This way, given an element $a$ in a finite monoid $\alg M$, we write $d_a$ for the only idempotent element that can be expressed as $a^n$ for some $n\in \NN$.

\begin{restatable}{lemma}{regsubmonoid}
\label{le:regular_submonoid}
Let $\alg M$ be a finite commutative monoid, $M_I\subseteq M$ its subset of idempotent elements, and $M_\dagger\subseteq M$ its subset of group elements. The following hold:
(1) The map $\pi_I: \alg M \to \alg M_I$ given by $a\mapsto d_a$
is a surjective monoid homomorphism satisfying $\pi_I\circ \pi_I= \pi_I$. 
(2) $\alg M_\dagger$ is a regular submonoid of $\alg M$. 
(3) The map $\pi_\dagger: \alg M \to \alg M_\dagger$ given by  
    $a\mapsto d_a a$ is a surjective monoid homomorphism satisfying $\pi_\dagger \circ \pi_\dagger = \pi_\dagger$.
\end{restatable}

Given a finite commutative monoid $\alg M$ and an element $a\in M$, we write $a_\dagger$ for the product $ad_a$ and, if $S\subseteq M$, we write $S_\dagger$ for $\{a_\dagger\mid a\in S\}$.

\begin{restatable}{lemma}{splitaux}
    \label{le:split_aux}
    Let $\alg M$ be a monoid, and $\alg N$ be a commutative monoid. Let $F$ be a finite family of 
    monoid homomorphisms $f: \alg M \to \alg N$, and $g: \alg M \to \alg N$ be the homomorphism given by $g(a)= \prod_{f\in F} f(a)$.\footnote{Observe that this notation is well-defined because products commute in a commutative monoid.} Then for each $S\subseteq \alg M$, it holds that $ 
    \Cos{g(S)_\dagger} \subseteq \bigotimes_{f\in F} \Cos{f(S)_\dagger}$. 
\end{restatable}

\section{Hardness}
\label{sec:hardness}

We prove here the hardness side of~\Cref{thm:main}.
That is, we show that in the absence of a suitable homomorphism $\rel M\to\rel
N$ we have that $\pcsp(\rel M,\rel N)$ is NP-hard.
For this, we first introduce an additional decision problem from~\cite{BBKO21} whose NP-hardness
can be proved under a certain algebraic condition $(\star)$, then exhibit a reduction from this problem to $\pcsp(\rel M,\rel N)$ which holds beyond the setting of groups and monoids, and finally show that $(\star)$ holds for the templates under consideration in this paper. 

\paragraph*{Polymorphisms}
Given a relational structure $\rel A$ and a number $n\in \NN$, the $n$-th Cartesian power of $\rel A$, denoted $\rel A^n$, is defined in the natural way. That is, the universe of $\rel A^n$ is $A^n$, and for each symbol $R$ in the signature, $R(\rel A^n)$ is the set of tuples $(\bm{a}_1, \dots, \bm{a}_{\arty(R)})$ where $\bm a_i\in R(\rel A)$ for each $i\in [\arty(R)]$. \par

Let $\rel{A} \to \rel{B}$ be two relational structures. Given $n\in \NN$,
a $n$-ary \emph{polymorphism} of the pair $(\rel A, \rel B)$ is a homomorphism $f: \rel A^n \to \rel B$.
We denote by $\pol^{(n)}(\rel A, \rel B)$ the set of $n$-ary polymorphisms of $(\rel A, \rel B)$, and use $\pol(\rel A, \rel B)$ to represent the disjoint
union $\bigsqcup_{n\in \NN} \pol^{(n)}(\rel A, \rel B)$.

Given $n,m\in \NN$, $f\in \pol^{(n)}(\rel A, \rel B)$, and $\sigma: [n]\to [m]$ the element $f^\sigma\in \pol^{(m)}(\rel A, \rel B)$ is given by $f^\sigma(\bm a) = f(\bm a\circ\sigma)$ for each $\bm a\in A^m$. In this situation we say that $f^\sigma$ is a \emph{minor} of $f$.
The \emph{unary minor} of $f$ is the polymorphism $f^{\tau}$, where $\tau$ is the constant map from $[n]$ to $[1]$.
In this situation, it is easy to see that
$(f^{\sigma_1})^{\sigma_2}= f^{\sigma_2\circ \sigma_1}$
for all suitable maps $\sigma_1,\sigma_2$,
and that $f^{\sigma}= f$
whenever $f\in \pol^{(n)}$, and $\sigma$ is the identity over $[n]$ for some $n\in \NN$. This sort of algebraic structure has been called a \emph{minion} in the literature~\cite{BBKO21}.

\paragraph*{Minor Conditions} A minor condition $\Phi$ is a tuple $(U,V, E, (\phi_{(u,v)\in E}))$ where
$U,V$ are finite disjoint sets where each element $u\in U\sqcup V$ has an arity $\arty(u)\in \NN$ associated to it, $E\subseteq U\times V$, and
$\phi_{u,v}$ is a map from $[\arty(u)]$ to $[\arty(v)]$
for each $(u,v)\in E$. We say that $\Phi$ is \emph{trivial} if for each $u\in V\sqcup U$ there is an index $i_u$ satisfying that 
$\phi_{u,v}(i_u)= i_v$ for each $(u,v)\in E$. Given relational structures $\rel A \to \rel B$, we say that $\Phi$ is \emph{satisfiable} over $\pol(\rel A, \rel B)$ if for each
$u\in V\sqcup U$ there is an element $f_u\in \pol^{(\arty(u))}(\rel A, \rel B)$ 
and there exists $f_v\in \pol^{(\arty(v))}(\rel A, \rel B)$
satisfying $f_u^{\phi_{u,v}}= f_v$ for each $(u,v)\in E$. It can be seen that if $\Phi$ is trivial, then it is satisfiable over $\pol(\rel A, \rel B)$ for any choice of $\rel A\to \rel B$. \par

Let $\rel A,\rel B$ be relational structures such that $\rel A\to\rel B$,
and let $\ell\in\NN$.
The \emph{promise minor condition problem}, denoted by $\PMC_\ell(\rel A, \rel B)$,\footnote{For the sake of readability, we differ slightly here from the notation introduced in~\cite{BBKO21}.}
is the computational problem of distinguishing, given a finite minor condition $\Sigma$ with symbols of arity at most $\ell$, between the following two cases: (1) $\Sigma$ is trivial, and (2) $\Sigma$ is not satisfiable in $\pol(\rel A, \rel B)$. 

Let $\rel A \to \rel B$ be relational structures, and $\mathscr{N}\subseteq \pol(\rel A, \rel B)$ be a subset.
A \emph{selection function} over $\mathscr{N}$ is a function $\mathcal I$ with domain $\mathscr{N}$ such that for every $n\geq 1$, and every $f\in \mathscr{N}$ of arity $n$, $\mathcal I(f)$ is a subset of $[n]$, and such that whenever $f\in \pol^{(n)}(\rel A, \rel B) \cap \mathscr{N}$ and $\sigma\colon [n]\to [m]$ are such that $f^\sigma\in \mathscr{N}$, then $\sigma(\mathcal I(f))\cap \mathcal I(f^\sigma)\neq\emptyset$.
We say that $\mathcal I$ is bounded if $\{|\mathcal I(f)| : f\in
\mathscr{N}\}\subseteq\mathbb N$ is bounded. The proof 
of
Theorem~5.21 in~\cite{BBKO21} implies the following result (cf. also~\cite{Barto22:soda}
and~\cite[Corollary~4.5]{KO22:survey}).

\begin{theorem}[\cite{BBKO21}] \label{hardness-criterion}
    Let $\rel A \to \rel B$ be relational structures.
    Suppose that
    \begin{equation*}\tag{$\star$}\parbox{0.9\textwidth}{$\pol(\rel A, \rel B)$ is the union of subsets $\Mscr_1,\dots,\Mscr_k$ such that
    for each $i\in [k]$ there exists a bounded selection function for $\Mscr_i$.}\end{equation*} Then there exists some $\ell\in \mathbb N$ such that  $\PMC_{\ell^\prime}(\rel A, \rel B)$ is NP-hard for all $\ell^\prime \geq \ell$.
\end{theorem}

\paragraph*{Algebraic Approach in the Finitely Generated Case}
For \emph{finite} relational structures, the problem $\PMC_\ell(\rel A,\rel B)$ is
related to $\pcsp(\rel A,\rel B)$ by the following result, which lies at the heart of the so-called algebraic approach to promise constraint satisfaction.
\begin{theorem}[\cite{BBKO21}]\label{equivalence-PMC}
    Let $\rel A \to \rel B$ be finite relational structures.
    Then $\PMC_\ell(\rel A,\rel B)$ and  $\pcsp(\rel A,\rel B)$ are equivalent under logspace reductions for all large enough $\ell$.
\end{theorem}

Let $\alg M, \alg N$ be monoids, where $\alg M$ is finitely generated and $\alg
N$ is finite, and $\rel{M}, \rel{N}$ be expansions of $M,N$ by a single relation  $R(M), R(N)$ such that $\rel{M} \rightarrow \rel{N}$.
The hardness side of \Cref{thm:main} is shown by reducing $\PMC_\ell(\rel M, \rel N)$ to $\pcsp(\rel M, \rel N)$ 
and then showing that $\PMC_\ell(\rel M, \rel N)$ is NP-hard for $\ell$ sufficiently large.
The fact that $\rel M$ may be infinite means that \Cref{equivalence-PMC} does not yield the reduction we are looking for.
In general, \Cref{equivalence-PMC} is known to fail for infinite structures, since $\PMC_L(\rel A,\rel B)$ is always in NP (since the computationally harder problem to check whether a given minor condition $\Sigma$ is trivial or not is itself in NP), while the (P)CSP of infinite structures can have arbitrary high complexities, even under strong structural assumptions~\cite{Gillibert22:sicomp}. 

Our reduction is more general and applies to a wider context than for groups and monoids.
In order to present it, we will need the standard notions from universal algebra
to put our result in its natural habitat.
An \emph{algebraic signature} $\tau$ is a list of function symbols, each having a finite arity $n\in\mathbb N$.
As for relational structures, we assume in this work that algebraic signatures are finite.
A \emph{$\tau$-algebra} $\alg A$ is a tuple consisting of a set $A$ together with, for each $n$-ary operation symbol $f\in\tau$, an operation $f^{\alg A}\colon A^n\to A$. The operations $f^{\alg A}$ are called the basic operations of the algebra.
A \emph{term} (in the signature $\tau$) is a formal operation obtained by composing the symbols in $\tau$ in a way that respects the arities of the symbols.
Every term $t$ can be evaluated naturally in any $\tau$-algebra, yielding an operation $t^{\alg A}$.
For example, if $\tau$ consists of a single binary symbol $\cdot$, then $t_1(x,y,z) = (x\cdot y)\cdot z$ and $t_2(x,y,z)=x\cdot(y\cdot z)$ are both terms. Any set with a binary operation $\alg A=(A,\cdot^{\alg A})$ is a $\tau$-algebra. We have $t_1^{\alg A}=t_2^{\alg A}$ in any algebra where $\cdot^{\alg A}$ is associative, e.g., in monoids and groups.

We can treat any algebra as a relational structure, by replacing a basic operation of arity $n$ by a relation of arity $n+1$.
The notion of homomorphism between algebras can be defined as homomorphisms between the associated relational structures.
Alternatively, a homomorphism between two $\tau$-algebras $\alg A$ and $\alg B$ is a function $h\colon \alg A\to \alg B$ such that 
\[ h(f^{\alg A}(a_1,\dots,a_n)) = f^{\alg B}(h(a_1),\dots,h(a_n))\] holds for every symbol $f$ of arity $n$ in $\tau$ and every $a_1,\dots,a_n\in \alg A$.
By a simple induction, one sees that this equality must also hold for every term $t$.

For $\ell\geq 1$, the $\ell$th \emph{power} of $\alg A$, denoted by $\alg A^\ell$, is the algebra with base set $A^\ell$ and whose fundamental operations are defined component-wise, i.e., by
\[ f^{\alg A^\ell}(\bm a^1,\dots,\bm a^n) = (f^{\alg A}(a^1_1,\dots,a^n_1),\dots,f^{\alg A}(a^1_L,\dots,a^n_\ell)),\]
where $a^i_j$ denotes the $j$th component of the $L$-tuple $\bm a^i$.
This coincides with the definition of products for groups and monoids.

We say that an algebra $\alg A$ is \emph{finitely generated} if there exists a
finite subset $A'=\{a_1,\dots,a_r\}\subseteq A$, called a \emph{generating set}, such that every
$a\in\alg A$ is of the form $t^{\alg A}(a_1,\dots,a_r)$ for some term $t$.

In general, an algebra $\alg A$ can be finitely generated while its powers are not. A simple example is the algebra with a single unary function $\alg N=(\mathbb N;s^{\alg N})$ where $s^{\alg N}(n)=n+1$.
One sees that $\{1\}$ is a generating set of $\alg N$, while any generating set of $\alg N^2$ must contain the elements $(1,1),(1,2),(1,3),\dots$.
However, if $\alg A$ is a finitely generated group (or a monoid) with identity element $e$, then $\alg A^\ell$ is finitely generated for all $L\in\mathbb N$.
Indeed, if $G=\{g_1,\dots,g_r\}$ is a generating set of $\alg A$, then one sees that elements of the form $(e,\dots,e,g_i,e,\dots,e)$ for all $i$ and all possible positions of $g_i$ generate $\alg A^\ell$.

Finally, we note the following. Suppose that $\alg A,\alg B$ are $\tau$-algebras and $U=\{\alpha_1,\dots,\alpha_r\}$ is a generating set of $\alg A$.
Let $h\colon U\to B$ be a function.
Then $h$ can be extended to a homomorphism $\tilde h\colon\alg A\to\alg B$ if, and only if, for all terms $s,t$ of arity $r$, we have that if $s^{\alg A}(\alpha_1,\dots,\alpha_r)=t^{\alg A}(\alpha_1,\dots,\alpha_r)$, then $s^{\alg B}(h(\alpha_1),\dots,h(\alpha_r))=t^{\alg B}(h(\alpha_1),\dots,h(\alpha_r))$.
Indeed, we can simply define $\tilde h(s^{\alg A}(\alpha_1,\dots,\alpha_r))$ as $s^{\alg B}(h(\alpha_1),\dots,h(\alpha_r))$, which is well-defined by assumption and is a total function by assumption that $G$ generates $\alg A$.

\begin{proposition}\label{prop:reduction-fin-presentation-general}
    Let $\alg A,\alg B$ be algebras such that all finite powers of $\alg A$ are finitely generated and $\alg B$ is finite. Fix $m\in \NN$ and $R^A\subseteq A^m, R^B\subseteq B^m$ such that $\rel A=(\alg A,R^A)$ admits a homomorphism to $\rel B=(\alg B,R^B)$. Then there is a logspace reduction from $\mathrm{PMC}_\ell(\rel A,\rel B)$ to $\pcsp(\rel{A}, \rel{B})$ for any $\ell \in \NN$.
\end{proposition}
 \begin{proof}
    Consider an instance $\Sigma$ of $\PMC_\ell(\rel A,\rel B)$. 
    Without loss of generality, we can assume that all symbols in
    $\Sigma$ have arity $\ell$ (otherwise lower-arity symbols are padded
    by dummy variables). We define an instance $\I= \I_\Sigma$ of $\pcsp(\rel A, \rel B)$ that is constructible in logspace from 
    $\Sigma$.  Let $U$ be a finite generating set of $\alg A^L$.
    We write $U=\{\bm \alpha_1,\dots,\bm \alpha_r\}$. 
    Without loss of generality, we can assume that for every $\bm \alpha\in U$ and $\sigma\colon[L]\to [L]$, we have $\bm \alpha \circ \sigma\in U$ as well.
    
    For each $x\in X$ and $\bm \alpha\in U$, let $x(\bm \alpha)$ be a variable of $\I$.
    Our goal is to obtain that in every solution of $\I$ in $\rel B$, the values $x(\bm \alpha)$ define a map $U\to \alg B$ that extends to a homomorphism $\tilde x\colon \alg A^L\to \alg B$.
    Note that there are only finitely many functions $f\colon U\to\alg B$.
    For each such function that does not extend to a homomorphism $\tilde f\colon\alg A^L\to\alg B$, there must exist a pair $(s,t)$ of terms such that $s^{\alg A^L}(\bm \alpha_1,\dots,\bm \alpha_r) = t^{\alg A^L}(\bm \alpha_1,\dots,\bm \alpha_r)$ while $s^{\alg B}(f(\bm \alpha_1),\dots,f(\bm \alpha_r))\neq t^{\alg B}(f(\bm \alpha_1),\dots,f(\bm \alpha_r))$.
    Make a finite list $(s_1,t_1),\dots,(s_k,t_k)$ of such pairs of terms corresponding to maps $U\to \alg B$ that do not extend to a homomorphism.
    For each such pair, and each variable $x$, add the constraint 
    \[ s(x(\bm \alpha_1),\dots,x(\bm \alpha_r)) = t(x(\bm \alpha_1),\dots,x(\bm \alpha_r))\]
    to $\I$.
    Formally, for each variable $x$ and pair $(s,t)$, this involves the introduction of a bounded number of variables for each subterm of $s$ and $t$ and corresponding equations.
    All in all, in every solution to $\I$, we are guaranteed that $x\colon U\to\alg B$ extends to a homomorphism $\tilde x\colon\alg A^L\to \alg B$.

    We now additionally want to make sure that the extension $\tilde x\colon\alg A^L\to \alg B$ of $x$ is a homomorphism $\rel A^L\to \rel B$, i.e., that it is an element of $\pol(\rel A,\rel B)$.
    Note that there are only finitely many homomorphisms $\alg A^L\to\alg B$.
    Let $f_1,\dots,f_k$ be those homomorphisms that are not homomorphisms $\rel A^L\to\rel B$.
    For each $i\in [k]$, proceed as follows.
    First, pick elements $\bm{r}^1,\dots,\bm{r}^L\in R^A$ such that $f_i(\bm{r}^1,\dots,\bm{r}^L)\not\in R^B$.
    Each $(r_j^1,\dots,r_j^L)$ for $j\in [r]$ is equal to $t_j^{\alg A^L}(\bm \alpha_1,\dots,\bm \alpha_r)$ for some term $t_j$.
    Add to $\I$ the constraint
    \[ (t_1(x(\bm \alpha_1),\dots,x(\bm \alpha_r)),\dots,t_{\arty(R)}(x(\bm \alpha_1),\dots,x(\bm \alpha_r))) \in  R. \]

    Finally, we want that the resulting homomorphisms satisfy the minor condition $\Sigma$.
    Let $x=y^\sigma$ be a minor identity in $\Sigma$.
    We add the constraints
    \[x(\bm \alpha) = y(\bm \alpha \circ \sigma)  \]
    for every $\bm \alpha\in U$, where we recall that $\bm \alpha \circ \sigma$ is also an element of $U$.
    We prove that if $\tilde x$ and $\tilde y$ are the resulting homomorphisms $\alg A^L\to\alg B$, then $\tilde x(a_1,\dots,a_\ell) = \tilde y(a_{\sigma(1)},\dots,a_{\sigma(N)})$ holds for all $a_1,\dots,a_\ell\in\alg A$.
    There exists a term $t$ such that $t^{\alg A^L}(\bm \alpha_1,\dots,\bm \alpha_r) = (a_1,\dots,a_\ell)$.
    Thus, $t^{\alg A}(\alpha_{1i},\dots,\alpha_{ri})=a_i$ for all $i\in [L]$.
    It follows that $t^{\alg A}(\alpha_{1\sigma(i)},\dots,\alpha_{r\sigma(i)})=a_{\sigma(i)}$ for all $i\in [L]$ and therefore $t^{\alg A^\ell}(\bm \alpha_1\circ\sigma,\dots,\bm \alpha_r\circ\sigma)=(a_{\sigma(1)},\dots,a_{\sigma(N)})$.
    By definition of $\tilde x$ and $\tilde y$, we have \[\tilde x(a_1,\dots,a_\ell) = t^{\alg B}(x(\bm \alpha_1),\dots,x(\bm \alpha_r))\] and \[\tilde y(a_{\sigma(1)},\dots,a_{\sigma(N)})=t^{\alg B}(y(\bm \alpha_1\circ \sigma),\dots,y(\bm \alpha_r\circ\sigma)).\]
    Since $\I$ contains the constraint $x(\bm \alpha_i)=y(\bm \alpha_i\circ\sigma)$ for all $i\in [r]$, we obtain that
    $\tilde x(a_1,\dots,a_\ell)=\tilde y(a_{\sigma(1)},\dots,a_{\sigma(N)})$ is satisfied.

    It remains to prove that this is a correct reduction.
    We have already described above that if $\I$ is not a \No-instance of
    $\pcsp(\rel A,\rel B)$, i.e., if there is a solution $\xi$ to
    $\I$, then for each $x\in X$ the function $U\to B$ defined by $\bm \alpha\mapsto \xi(x(\bm\alpha))$ extends to an element $\tilde{x}\in\pol(\rel A, \rel B)$ showing
    that $\Sigma$ is satisfiable in $\pol(\rel A,\rel B)$ and therefore not a \No-instance of $\PMC_\ell(\rel A,\rel B)$.
    Therefore, \No-instances of $\PMC_\ell(\rel A,\rel B)$ are mapped to \No-instances of $\pcsp(\rel A,\rel B)$.   Conversely, suppose that $\Sigma$ is a \Yes-instance of $\PMC_\ell(\rel A,\rel B)$. Then there is a map $x \mapsto i_{x}$ 
    that satisfies $\sigma(i_x)=i_y$ for any constraint of the form $x^\sigma=y$ in $\Sigma$. Then, the map that sends each variable $x(\bm \alpha_j)$
    to the $i_x$-th component of $\bm \alpha_j$ is a solution over $\rel A$
    to our instance.

\end{proof}

\paragraph*{Polymorphisms of Expansions of Monoids}

For the NP-hardness part of~\Cref{thm:main}, we are left to show in this section that if there does not exist a homomorphism $h\colon\rel M\to\rel N$ whose image is a {regular} commutative submonoid of $\alg N$ and such that $\Cos{h(R(M))}\subseteq R(N)$, then the requirement in~\Cref{hardness-criterion} is met. \par
The following two auxiliary results show that in a finite commutative monoid $\alg M$, for a sufficiently large $\ell$ and any set $U\subseteq M$,
any product of the form $\prod_{i=1}^\ell a_i$,
where $a_i\in \Cos{U_\dagger}$ for each $i$, can be obtained alternatively as $\prod_{i=1}^\ell b_i$, where $b_i\in U$ for each $i$.

\begin{lemma}\label{splitting-monoids}
	Let $\alg M$ be a finite {regular} commutative monoid. There exists $\ell(\alg M)\in\mathbb N$ such that for all $n \geq \ell(\alg M)$ and for all $U\subseteq \alg M$ we have $\Cos{U}^{\otimes n} = U^{\otimes n}$.
\end{lemma}
\begin{proof}  
Recall the decomposition of {regular} commutative monoids described in \Cref{sec:algebra}. Let $M_I$ be the set of idempotent elements in $\alg M$, and $H_d$ the maximal subgroup of $\alg M$ containing $d$ for each $d\in M_I$. \par
For all $n$, we have $|U^{\otimes n}|\leq |U^{\otimes n+1}|$. To see this, 
we simply show that $|U^{\otimes n} \cap H_d| \leq
|U^{\otimes n+1} \cap H_d|$ for each $d\in M_I$.
Suppose that $|U^{\otimes n} \cap H_d|\geq 1$ (otherwise we are done).
Then there must be some element $a\in U$ with $d\preceq a$.
It follows that $ad \Hcal d$, so $ad\in H_d$. This way,
\[
a\otimes H_d = a \otimes (d \otimes H_d) = ad \otimes H_d \subseteq H_d,
\]
and therefore $a\otimes (U^{\otimes n} \cap H_d)\subseteq U^{\otimes n+1}\cap H_d$.
Moreover, $|a\otimes (U^{\otimes n} \cap H_d)| = |U^{\otimes n} \cap H_d|$ holds, proving the statement.
  Thus, since $\alg M$ is finite, there exists a natural number $\ell(M)$ such that $|U^{\otimes n}|=|U^{\otimes \ell(M)}|=\lambda$ holds for all $n\geq \ell(\alg M)$.

Now let us show that $|U\otimes (U^{-1}\otimes U)^{\otimes n}| \geq |U^{\otimes n}|$ for all $n\in \NN$. To do this we define an injective map
from $U^{\otimes n} \cap H_d$ to $U\otimes (U^{-1}\otimes U)^{\otimes n} \cap H_d$ for each $d\in M_I$. Suppose that $U^{\otimes n} \cap H_d$ is not empty. In particular, there must be some element $a\in U$ satisfying $d\preceq a$. We claim the map $b \mapsto a^{1-n} b$ is an injective map from $U^{\otimes n} \cap H_d$ to $U\otimes (U^{-1}\otimes U)^{\otimes n} 
\cap H_d$. 
 This is a well-defined map: if $b\in U^{\otimes n} \cap H_d$, then we can write $b=s_1 \cdots s_n$ for some $s_1,\dots, s_n \in U$, and then $a^{1-n}b= 
a(a^{-1}s_1)\cdots (a^{-1}s_n)$, which belongs to $U\otimes (U^{-1}\otimes U)^{\otimes n}$. The fact that $b\in H_d$ and $b\preceq a$ implies $a^{1-n} b\in H_d$. Finally, let us show that the map $b \mapsto a^{1-n} b$ is injective. 
Observe that $b=d b$, for all $b\in H_d$, so $a^{1-n} b = a^{1-n} b^\prime$ if,
  and only if, $(a^{1-n}d) b = (a^{1-n}d) b^\prime$ for each
$b,b^\prime \in U\otimes(U^{-1}\otimes U)^n \cap H_d$. However, the fact that
  $d\preceq a$ implies that $a^{1-n}d \in H_d$, so  $(a^{1-n}d) b = (a^{1-n}d)
  b^\prime$ holds if, and only if, $b = b^\prime$. \par
Now we can see that $\Cos{U}=U^{\otimes n}$ for some $n$. Indeed, $\Cos{U}=U \otimes (U^{-1}U)^{\otimes n}$ for all large enough $n$, so in particular $|\Cos{U}|\geq |U^{\otimes n}| = \lambda$. Since $U^{-1}\subseteq U^{\otimes k_M-1}$ for some constant $k_M$, we get $\Cos{U}\subseteq U^{\otimes (1+n k_M)}$ and this set has size at most $\lambda$, so the two sets are equal.\par
Finally, for all $n$, we trivially have $U^{\otimes n}\subseteq \Cos{U}^{\otimes n}$ and since for $n\geq \ell(\alg M)$ the two sets have size $\lambda$, we get $U^{\otimes n}=\Cos{U}^{\otimes n}$.
\end{proof}

\begin{lemma}\label{splitting-monoids_non_regular}
Let $\alg M$ be a finite commutative monoid. There exists $\ell(\alg M)\in\mathbb N$ such that for all $n \geq \ell(\alg M)$ and for all $U\subseteq \alg M$ we have $\Cos{U_\dagger}^{\otimes n} \subseteq U^{\otimes n}$.
\end{lemma}
\begin{proof}
    Let $k\in \NN$ be such that $a^k$ is idempotent for all $a\in \alg M$, as given by \Cref{le:idempotent_constant}. Let $\lambda$ be a constant witnessing \Cref{splitting-monoids} with respect to $\alg M_\dagger$. Let $\ell = \max( \lambda, 2k|M|)$. We claim that $\ell$ satisfies the lemma's statement. Fix $n\geq \ell$. Then, it is enough to show that $\Cos{U_\dagger}^{\otimes n} \subseteq U^{\otimes n}$.
    
    Let $a \in \Cos{U_\dagger}^{\otimes n}$. By \Cref{splitting-monoids}, $a= \prod_{i\in [n]} (b_i)_\dagger$ for some $b_1,\dots, b_n\in U$. Another way of writing this identity is
    $a = \prod_{b\in U^\prime} b_\dagger^{n_b}$, where 
    $U^\prime \subseteq U$, and the $n_b$ denote positive integers satisfying $\sum_{b\in U^\prime} n_b= n$. As $n \geq 2k|M|$, 
    and $|U^\prime|\leq |M|$, there are
    integers $(m_b)_{b\in U^\prime}$
    such that $\sum_{b\in U^\prime} m_b = n$ that satisfy both  $m_b  = n_b \mod k$, and $m_b \geq k$ for all $b\in U^\prime$. Let us construct this sequence. For each $b\in U^\prime$, let $0 < n^\prime_b \leq k$ be the only integer satisfying $n^\prime_b = n_b \mod k$. Define $m^\prime_b= n^\prime_b+k$ for each $b\in U^\prime$. Define 
    \[\tau= \sum_{b\in U^\prime} m^\prime_b - \sum_{b\in U^\prime} n_b.\] 
    As $n\geq 2k|M|$ and $m^\prime_b\leq 2k$ for each $b\in U^\prime$, we have
    that $\tau \leq 0$. Moreover, by construction $m^\prime_b = n_b \mod k$ for each $b\in U^\prime$. This way, $\tau = 0 \mod k$. In order to obtain the desired $(m_b)_{b\in U^\prime}$, define $m_b= m^\prime_b + \tau$ for one $b\in U^\prime$, and $m_{c}= m^\prime_c$
    for all $c\in U^\prime \setminus \{ b\}$. 
    This way,
    \[
    a = \prod_{b\in U^\prime} b_\dagger^{n_b}= \prod_{b\in U^\prime} b_\dagger^{m_b} = \prod_{b\in U^\prime} b^{m_b}.
    \]
    The last term belongs to $U^{\otimes n}$, so this completes the proof. 
\end{proof}

With this in hand, we are now able to provide a bounded selection function for the subset of the polymorphisms of $(\rel M,\rel N)$
with commutative image, and whose unary minor has a regular image.

Let $\alg M$ be a monoid with identity $e$. We say that two sets, $S, T\subseteq \alg M$ \emph{commute} if 
$ab = ba$ for every $a\in S, b\in T$.  Let $\alg M, \alg N$ be monoids and $f\colon \alg M^n\to \alg N$ be a homomorphism.
For $i\in [n]$, we define the homomorphism $f_i\colon \alg M\to \alg N$ by $f_i(x) = f(e,\dots,e,x,e,\dots,e)$, where $x$ is at position $i$.
This way, $f(x_1,\dots,x_n)=\prod_{i=1}^n f_i(x_i)$. Moreover, $\im(f_i)$ and $\im(f_j)$ commute for each $i\neq j$. 
We say that $i,j\in[n]$ are \emph{$f$-equivalent} if $f_i=f_j$. This clearly defines an equivalence relation on $[n]$ whose equivalence classes are called \emph{$f$-equivalence classes}. If $I$ is an $f$-equivalence class,
  we define $f_I$ as the homomorphism $f_i$ for an arbitrary $i\in I$.

\begin{observation}
\label{le:minoring}
    Let $\alg M, \alg N$ be monoids, $f\colon \alg M^n \rightarrow \alg N$ a homomorphism, and $g= f^\sigma$ for some $\sigma:[n]\to [m]$. Then, 
    $g_i(a) = \prod_{j\in \sigma^{-1}(i)} f_j(a)$
    for each $i\in [m]$, $a\in M$.
\end{observation}

\begin{observation}
\label{obs:pol}
    Let $\rel{M}, \rel{N}$ be expansions of monoids $M,N$ by a single relation
    $R(M), R(N)$ such that $\rel{M} \rightarrow \rel{N}$. Then a map $p\colon M^n
    \rightarrow N$ is a polymorphism in $\pol(\rel{M}, \rel{N})$ if, and only
    if, $p$ is a monoid homomorphism from $M^n$ to $N$ and $p((R(M))^n)\subseteq p(R(N))$.
\end{observation}

\begin{proposition}\label{prop:minor-compatibility_coset}
  Let $M, N$ be monoids with $\alg M \to \alg N$, where $\alg N$ is finite. Fix $r\in \NN$ and relations $R\subseteq M^r, S\subseteq N^r$. 
Let $\Mscr \subseteq \pol(\alg M, \alg N)$ be the subset of polymorphisms $f$ such that (1) $\im(f)$ is commutative, (2)
the unary minor $h$ of $f$ has a {regular} image, and (3) 
$\Cos{h(R)}\not\subseteq S$. 
Given $f\in \Mscr$, 
let $\mathcal{I}(f)\subseteq[n]$ be the union of all $f$-equivalence classes of size smaller than $\ell(\alg N^r)$,
where $\ell$ is given by \Cref{splitting-monoids_non_regular}. Suppose that
$\mathcal I(f^\sigma)\cap\sigma(\mathcal I(f))= \emptyset$ for some $f, f^\sigma\in \Mscr$. Then $f(R^n)\not\subseteq S$, where $n$ denotes the arity of $f$.
\end{proposition}
\begin{proof}
Let $f, g\in \Mscr$ be polymorphisms satisfying $f^\sigma= g$, and
$\mathcal I(g)\cap\sigma(\mathcal I(f))= \emptyset$. Suppose the arities of $f$ and $g$ are $n$ and $m$ respectively.  Let $J_1,\dots,J_p\subseteq [m]$ be a list of the $g$-equivalence classes.
By assumption $\Cos{h(R)}\not \subseteq S$, where $h$ is the unary minor of $g$ and $f$. Let $\bm c\in \Cos{h(R)} \setminus S$.
By \Cref{le:split_aux}, there exist elements
$\bm b_i\in \Cos{g_i(R)_\dagger}$ for each $i\in [m]$ satisfying 
$\prod_{i\in [m]} \bm b_i = \bm c$.

In this paragraph, we construct a second sequence $(\bm b^\prime_i)_{i\in [m]}$ whose product is also $\bm c$, but where $\bm b^\prime_i\in g_i(R)$ for each $i\notin \mathcal{I}(g)$.
First, we define $\bm b^\prime_i= \bm b_i$ for each $i\in \mathcal{I}(g)$.
Now let $J\not \subseteq \mathcal{I}(g)$ be a $g$-equivalence class. Then
\[
\prod_{j\in J} \bm b_j\in \Cos{g_J(R)_\dagger}^{\otimes |J|}.\]
We know that $|J| \geq \ell(\alg N^r)$, so by \Cref{splitting-monoids_non_regular} we conclude there are elements $\bm b^\prime_i\in g_J(R)$ for each $i\in J$ satisfying 
$\prod_{j\in J} \bm b_j \bm d_{\bm b_j} = \prod_{j\in J} \bm b^\prime_j$. Applying this reasoning to each $g$-equivalence class $J \not \subseteq \mathcal{I}(g)$, we define an element $\bm b^\prime_i\in g_i(R)$ for each $i\notin \mathcal{I}(g)$. This way,  $
\prod_{i \in [m]} \bm b^\prime_i = 
\bm c$. \par

We now construct another sequence $(\bm a_j)_{j\in [n]}$ whose product equals $\bm c$. For each $i\in \mathcal{I}(g)$, by \Cref{le:minoring} it holds that $g_i=\prod_{j\in \sigma^{-1}(i)} f_j$.
Thus, by \Cref{le:split_aux}, $\Cos{g_i(R)_\dagger}\subseteq \bigotimes_{j\in
  \sigma^{-1}(i)} \Cos{f_j(R)_\dagger}$. We use this fact to obtain elements $\bm a_j\in
  \Cos{f_j(R)_\dagger}$ for each $j\in \sigma^{-1}(i)$ satisfying $\prod_{j\in
  \sigma^{-1}(i)} \bm a_j = \bm b^\prime_i$. Similarly, for each $i\in [m]
  \setminus \mathcal{I}(g)$, we have that $\bm b^\prime_i= g_i(\bm \alpha)$ for
  some $\bm \alpha\in R$, so we define $\bm a_j= f_j(\bm \alpha)$ for each $j\in
  \sigma^{-1}(i)$. This way, $\prod_{j\in \sigma^{-1}(i)} \bm a_j = \bm
  b^\prime_i$ as well. Hence, we have obtained a sequence $(\bm a_j)_{j\in [n]}$
  satisfying $\prod_{j\in [n]} \bm a_j = \bm c$. Observe that, according to our
  construction, if $j\notin \sigma^{-1}(\mathcal{I}(g))$, then $\bm a_j\in
  f_j(R)$. Finally, we construct a fourth sequence $(\bm a^\prime_j)_{j\in [n]}$
  whose product equals $\bm c$ and where $\bm a^\prime_j\in f_j(R)$ for each
  $j\in [n]$. Given $j\in \mathcal{I}(f)$, we define $\bm a^\prime_j= \bm a_j$.
  Observe that by the assumption of the lemma, in this case $j\notin \sigma^{-1}(\mathcal{I}(g))$, so $\bm a^\prime_j\in f_j(R)$. Now let $J\not \subseteq \mathcal{I}(f)$ be an $f$-equivalence class. Then \[
\prod_{j\in J} \bm a_j \bm d_{\bm a_j}\in \Cos{f_j(R)_\dagger}^{\otimes(|J|)}.\]
By the definition of $\mathcal{I}$, we know that $|J| \geq \ell(\alg N^r)$, so by \Cref{splitting-monoids_non_regular} we conclude there are elements $\bm a^\prime_i\in f_J(R)$ for each $j\in J$ satisfying 
$\prod_{j\in J} \bm a_j \bm d_{\bm a_j} = \prod_{j\in J} \bm a^\prime_j$.
  Applying this reasoning to each $f$-equivalence class $J \not \subseteq
  \mathcal{I}(f)$, we define an element $\bm a^\prime_j\in f_j(R)$ for each $j\notin \mathcal{I}(f)$. 
This way, 
\[
\prod_{j\in [n]} \bm a^\prime_j= 
\left(\prod_{j\in \mathcal{I}(f)} \bm a_j \right) \left(\prod_{j\in [n]\setminus \mathcal I(f)} \bm a_j  \bm d_{\bm a_j} \right)
= \prod_{j\in [n]} \bm a_j = \bm c.\]
To see the second equality, observe that $\bm c$ is a group element, and $\bm c\preceq \bm a_j$ for all $j\in [n]$, so $\bm c \bm d_{\bm a_j} = \bm c$. Now, by construction $\bm a^\prime_j = f_j(\bm \alpha_j)$ 
for some $\bm \alpha_j\in R$, for each $j\in [n]$. In particular,
$f(\bm \alpha_1,\dots, \bm \alpha_n)=\bm c$. However, $f(\bm \alpha_1,\dots, \bm \alpha_n)\in f(R^n)$, and $\bm c\notin S$, so this proves the result.  
\end{proof}

We conclude this section proving the NP-hardness part of our main theorem
(\Cref{thm:main}).
\begin{proposition}
    \label{prop:hardness_main}
    Let $\alg M$ be a finitely generated monoid, $\alg N$ be a finite monoid,
    and $R(M)\subseteq M^r, R(N)\subseteq N^r$ be relations of arity $r$ satisfying that $\rel M=(\alg M, R(M))$ maps homomorphically to $\rel N= (\alg N, R(N))$. Suppose that there is no homomorphism $f:\rel M \to \rel N$ with regular commutative image such that $\Cos{f(R(M))}\subseteq R(N)$. Then $\pcsp(\rel M, \rel N)$ is NP-hard. 
\end{proposition}
\begin{proof}
    We show that $\PMC_k(\rel M,\rel N)$ is NP-hard for sufficiently large $k$. By \Cref{prop:reduction-fin-presentation-general} this proves the result. \par
    We describe $\pol(\rel M, \rel N)$ as the union of three subsets (1) the set $\Mscr_1$  of polymorphisms $f$ such that $\im(f)$ is not a commutative submonoid of $\alg N$, (2) the set $\Mscr_2$ of polymorphisms $f$ such that $\im(f)$ is commutative and whose unary minor $g$ satisfies that $\im(g)$ is not a {regular} submonoid of $\alg N$, and 
    (3) $\Mscr_3$, the complement of $\Mscr_1\cup\Mscr_2$ in $\Mscr$.

    Using that that $\alg N$ is finite, the proof of~\cite[Theorem 3]{LZ24:acm} on pages 3:14--3:15 therein shows the existence of a selection function for $\Mscr_1$ and $\Mscr_2$.
    It remains therefore to construct a selection function for $\Mscr_3$.

    Since there is no homomorphism $h\colon\alg M\to\alg N$ with regular commutative
    image and such that $\Cos{h(R(M))}\subseteq R(N)$, all the elements of $\Mscr_3$ satisfy the assumption of~\Cref{prop:minor-compatibility_coset}.
    Thus, by~\Cref{prop:minor-compatibility_coset}, the map $\mathcal I$ where for an $n$-ary $f\in\Mscr_3$ we define $\mathcal I(f)$ to be the union of all $f$-equivalence classes of size smaller than the constant $\ell(\alg N^r)$ given in \Cref{splitting-monoids_non_regular} is a selection function.
    It remains to show that it is bounded.
    Since $\alg M$ is finitely generated and $\alg N$ is finite, there is a finite number $\lambda$ of homomorphisms $\alg M\to \alg N$, and therefore for an $n$-ary $f\in\Mscr_3$ there are only $\lambda$ possibilities for $f_i$ with $i\in [n]$.
    Thus, the number of distinct $f$-equivalence classes is $\lambda$ for each $f\in\Mscr_3$. Therefore, at most $\lambda$ of those $f$-equivalence classes have size smaller than $\ell(\alg N^r)$.
    It follows that $\mathcal I$ is a bounded selection function, which concludes the proof.
\end{proof}

\section{Tractability}\label{sec:tractability}

Let $\rel{A} \to \rel{B}$ be two relational structures, and $i\in \NN$. A \emph{$2$-block 
symmetric polymorphism} $f\in \pol(\rel A, \rel B)$ of arity $2i+1$ is a 
homomorphism $f: \rel A^{2i+1} \to \rel B$ such that $f^\sigma=f$ for all bijections $\sigma:[2i+1]\to [2i+1]$ satisfying $\sigma([i+1])=[i+1]$ (i.e., $\sigma$ preserves the sets
$\{1,\dots, i+1\}$ and $\{i+2, \dots, 2i+1\}$). 
We use the following characterization of solvability of PCSPs via the BLP+AIP algorithm defined in~\cite{BGWZ20}.

\begin{theorem}[\cite{BGWZ20}]\label{th:BLP+AIP}
Let $\rel A \to \rel B$ be finite relational structures. Then $\pcsp(\rel A,
  \rel B)$ is solvable in polynomial time via BLP+AIP if, and only if,
  $\pol(\rel A, \rel B)$ contains $2$-block symmetric polymorphisms of all odd arities.     
\end{theorem}

The following shows the tractability side of our main result, \Cref{thm:main}.

\begin{proposition}
\label{prop:main_tractability}
     Let $\rel M=(\alg M,R(M)), \, \rel N=(\alg N,R(N))$ be expansions of monoids such that there exists a homomorphism $\rel M\to\rel N$. Assume that $\alg M$ is finitely generated and $\alg N$ is finite. The following are equivalent.
     \begin{itemize}
         \item[(1)] There is a homomorphism
         $h\colon \rel M \to \rel N$ with  {regular} commutative image, satisfying that
         $\Cos{h(R(M))}\subseteq R(N)$.
         \item[(2)] $\pol(\rel M, \rel N)$ contains $2$-block symmetric polymorphisms of all odd arities.
         \item[(3)] There is a finite structure $\rel A$
         satisfying $\rel M \to \rel A \to \rel N$ such that $\csp(\rel A)$ is solvable via BLP+AIP.
       \item[(4)] There is a homomorphism from $(\abreg{\alg M}, T)$ 
         to $\rel N$, where $T=\Cos{\abreg{\pi}_M(R(M))}$, and $\abreg{\alg M}, \abreg{\pi}_M$ are given by \Cref{le:abelian_regularization}.
     \end{itemize}
\end{proposition}
\begin{proof}
    We show that (1) is equivalent to each of the other items. The arguments are similar to those in~\cite{LZ24:acm}, but the implication (2)$\implies$(1) is significantly more complex.    
    \par
    \textbf{(1)$\implies$(2)} Let $i\in \NN$. We construct a $2i+1$-ary $2$-block symmetric polymorphism $f\in \pol(\rel A, \rel B)$. We define 
    \[f(a_1,\dots, a_{i+1}, b_1, \dots, b_{i}) = 
    \prod_{j=1}^{i+1} h(a_j)  \prod_{j=1}^{i} h(b_j)^{-1}
    \]
    for each $(a_1,\dots, a_{i+1}, b_1, \dots, b_{i})\in M^{2i+1}$. The fact that $f$ is a monoid homomorphism from $\alg M^{2i+1}$ to $\alg N$ follows from $\im(h)\leq \alg N$
    being commutative and {regular}. We only need to show that $f\left(R(M)^{2i+1}\right)\subseteq R(N)$. This can be seen as follows: 
    \begin{align*}
 & f\left(R(M)^{2i+1}\right) = h(R(M))^{\otimes i+1} \otimes (h\left(R(M))^{-1}\right)^{\otimes i} =
 \\
& h(R(M))\otimes (h(R(M))^{-1} \otimes h(R(M)))^{\otimes i} \subseteq \Cos{h(R(M))} \subseteq R(N).
\end{align*}
    
    \par
    \textbf{(2)$\implies$(1)} Suppose, for the sake of a contradiction, that there is no homomorphism $h: \rel M \to \rel N$
    satisfying (1). In \Cref{prop:hardness_main} we show that in this case
    $\pol(\rel M, \rel N)$ is in the scope of \Cref{hardness-criterion}. This
    means that $\pol(\rel M, \rel N)$ is the union of subsets
    $\Mscr_1,\dots,\Mscr_k$ such that there exists a bounded selection
    function $\mathcal{I}_i$ for $\Mscr_i$ for each $i\in [k]$. Let $\ell$ be a bound on $|\mathcal{I}_i(f)|$ for all $i\in [k], f\in \Mscr_i$. Let $f$ be a $2$-block symmetric polymorphism of arity $4\ell + 1$, and let $i$ be such that $f\in \Mscr_i$.
    Because $|\mathcal{I}_i(f)|\leq \ell$, and each block of $f$ has size at least $2\ell$, there is a bijection $\sigma: [4\ell+1]\to [4\ell+1]$ that preserves each block and satisfies 
   $\mathcal{I}_i(f) \cap \sigma(\mathcal{I}_i(f))=\emptyset$. However, $f=f^\sigma$, so this contradicts the fact that $\mathcal{I}_i$ is a selection function on $\Mscr_i$. We derived this contradiction from the assumption that (1) does not hold, so this completes the proof. \par
   \textbf{(1) $\implies$ (3)} Let $\rel A$ be the expansion of the monoid $\im(h)$ by the relation $\Cos{h(R)}$. Observe that $h$ is a homomorphism from $\rel M$ to $\rel A$ and the inclusion is a homomorphism from $\rel A$ to $\rel N$. We show that $\csp(\rel A)$ is solvable via BLP+AIP. By \Cref{th:BLP+AIP}, we just need to find a $2$-block symmetric polymorphism $f\in \pol(\rel A)$ of arity $2i+1$ for each $i\in \NN$. Here $\pol(\rel A)$ denotes $\pol(\rel A, \rel A)$. We define
   \[
   f(a_1,\dots, a_{i+1}, b_1,\dots, b_i) = \prod_{j=1}^{i+1} a_j \prod_{j=1}^i b_j^{-1},
   \]
   for each $a_1, \dots, a_{i+1}, b_1, \dots, b_i\in \im(h)$.
   The fact that $f$ is a well-defined monoid homomorphism from $\im(h)^{2i+1}$ to $\im(h)$ follows from the fact that $\im(h)$ is commutative and {regular}. Let $T= \Cos{h(R)}$. To see that $f(T^{2i+1})\subseteq T$, observe that 
   $f(T) = T^{\otimes i+1} \otimes (T^{-1})^{\otimes i}
   = T \otimes (T^{-1} \otimes T)^{\otimes i} = T$,
   where the last equality uses the fact that $T$ is a coset. \par
   \textbf{(3) $\implies$ (2)} Suppose there is a finite structure $\rel A$
   such that $\csp(\rel A)$ is solvable via BLP+AIP, and there are homomorphisms 
   $g_1: \rel M \to \rel A$ and $g_2: \rel A \to \rel N$. Let $i\in \NN$. We show that
   $\pol(\rel M, \rel N)$ has a $2$-block symmetric polymorphism $f$ of arity $2i+1$. By \Cref{th:BLP+AIP}, $\pol(\rel A)$ contains a $2$-block symmetric polymorphism $h$ of arity $2i+1$. Then we can define $f = g_2 \circ h \circ g_1$. %
    \textbf{(1) $\implies$ (4)} Follows directly from \Cref{le:abelian_regularization}.\par
    \textbf{(4) $\implies$ (1)} Follows from the fact that the homomorphic image of a 
     {regular} commutative monoid must be commutative and {regular}, and the homomorphic image of a coset is a coset.     
\end{proof}

\paragraph*{Algorithm for Infinite Templates}

Let $\rel M$ be the expansion of a finitely generated {regular} commutative monoid 
$\alg M$ by a coset $R(M)\subseteq M^r$. We sketch a polynomial-time algorithm solving $\csp(\rel M)$. This way, using item (4) in \Cref{prop:main_tractability}, we obtain that all the tractable problems $\pcsp(\rel A, \rel B)$ within the scope of our main theorem,
\Cref{thm:main},
can be solved by an algorithm that does not depend on the second structure $\rel B$ of the template. Our algorithm is an extension of the algorithm outlined in~\cite[Lemma 23]{KTT07:tcs}, which is a polynomial-time algorithm to solve systems of equations over a finite  {regular}  commutative monoid, although we observe the algorithm can also be applied to finitely generated monoids. The full details can be found in \Cref{ap:algorithm}
\par

By \Cref{prop:charRM}, we can describe $\alg M$ as the disjoint union
$\bigsqcup_{d\in M_I} H_d$, where $\alg M_I \leq \alg M$ is the semilattice formed by the idempotent elements in $\alg M$ and
$\alg H_d$ is the $\Hcal$-class of $d\in M_I$, which is a subgroup of $\alg M$. Using the fact that $\alg M$ is finitely generated,
\Cref{structure-abelian-monoid} yields that $\alg M_I$ is finite. We define
$\pi_I: \alg M \to \alg M_I$ as the natural projection that sends each element $a\in M$ to the only idempotent element $\delta_a$ in its $\Hcal$-class.
Finally, let $R(M_I)\subseteq M_I^r$ be the set $\{ \pi_I(\bm a) \vert  \bm a\in R(M) \}$, and $\rel M_I$ be the expansion of $\alg M_I$ by $R(M_I)$.  \par

Let $\rel X$ be an instance of $\csp(\rel M)$. We can see this instance as a set of variables $X$ and constraints of the form $x_1x_2=x_3$, $x=e$, and $(x_1, \dots, x_r)\in R$. If there is a solution $f: \rel X \to \rel M$, then $\pi_I\circ f$ is a solution for $\rel X$ in $\csp(\rel M_I)$. It is not difficult to see that homomorphisms $h: \rel X \to \rel M_I$ are closed under point-wise products. Hence, if there exists such a homomorphism, there must be a minimal homomorphism $h_{\mathrm{min}}$. ``Minimal'' here means that $h_{\mathrm{min}}(x)\preceq h(x)$ for any $h: \rel X \to \rel M_I$, $x\in X$. It can also be seen that if  $f:\rel X \to \rel M$ and $h: \rel X \to \rel M_I$ are homomorphisms, then their point-wise product $fh$ is still a homomorphism from $\rel X$ to $\rel M$. This way, we can conclude that if 
there is such a homomorphism $f$, then there must be one satisfying $\pi_I\circ f= h_{\mathrm{min}}$. Our algorithm aims to find a homomorphism $f$ with this property. \par

Firstly, arc-consistency (e.g., \cite{BBKO21}) or a slight modification of the
algorithm in {\cite[Lemma~19]{KTT07:tcs}} allows us to find $h=h_{\mathrm{min}}$ if $\rel X$ is a satisfiable instance. Using the homomorphism $h$ we reduce the problem of finding $f:\rel X \to \rel M$ to the problem of solving a system of linear equations $\Sigma_X$ over $\ZZ$, which can be solved in polynomial time using Gaussian elimination. The key insight is that for each $d \in M_I$, the group $\alg H_d$ is commutative and finitely generated, so it is isomorphic to the homomorphic image of some power of $\ZZ$ (e.g., \cite[Theorem 1.4, Chapter II]{hungerford2012algebra}). In other words, $\alg H_d$ is isomorphic to $\ZZ^{k_d}/G_{d}$ for some number $k_d\in \NN$ and some subgroup $G_d \leq \ZZ^{k_d}$. This way, we can choose a big enough number $k\in \NN$ such that all the elements in $\rel M$ can be represented with pairs of the form $(d, \bm v)$, where $d \in M_I$ is an idempotent element, and $\bm v\in \ZZ^k$
is an integer vector that is mapped to the quotient  $\ZZ^{k_d}/G_d$. We refer to \Cref{ap:algorithm} for the details. 
In $\Sigma_X$ we represent each variable $x\in X$ by a vector of integer variables $\bm v^x=(v^x_1, \dots, v^x_k)$. 
Constraints of the form $x=e$ are translated to $\bm v^x=\bm 0$. Given a constraint of the form $xy=z$ in $\rel X$, we add to $\Sigma_X$ equations ensuring $\bm v^x + \bm v^y = \bm v^z$ inside $\alg H_d$, where $d = h(z)$. Finally, we observe that for each $\bm{d}\in R(M_I)$, the intersection $R(M)\cap H_{\bm d}$ is a coset \footnote{We remind the reader that the notion of coset introduced in \Cref{sect:prelims} generalizes cosets in Abelian groups, so there should be no confusion here.}. A coset $U$ in a commutative group $\alg{G}$ can be expressed (using additive notation) as $a + H$, where
$a\in G$ is an element, and $\alg H\leq \alg G$ is a subgroup. This way, the condition $\bm a \in R(M_I)$ can be expressed via a set of linear equations inside each group of the form $\alg H_{\bm d}$ with $\bm d\in M_I^r$. Using this insight, we represent each constraint $(x_1,\dots, x_r)\in R$ in $\rel X$ with a set of linear equations in $\Sigma_X$ ensuring
$(v^{x_1},\dots, v^{x_r})\in R(M_I) \cap \alg H_{\bm d}$, where 
$\bm d= (h(x_1), \dots, h(x_r))$. This completes the reduction.

\section{Algorithm For Infinite Templates}
\label{ap:algorithm}

Let $\alg N$ be a semilattice with identity element $e(\alg N)$, $Q$ a set, and let $\lambda: N \to 2^Q$ be a map satisfying $\lambda(a) \subseteq \lambda(b)$ whenever $b\preceq a$, for any $a,b\in N$. The monoid $\alg M = \alg N \times_\lambda \ZZ$ consists of the set 
\[ \bigsqcup_{a\in N} \{a\}\times \ZZ^{\lambda(a)}. \] 
We implicitly identify $\ZZ^{\lambda(a)}$ with the subspace of $\ZZ^{Q}$
consisting of all vectors $\bm{v}$ satisfying $v_\alpha=0$ for all $\alpha\notin \lambda(a)$. This way, given $a,b\in N$, $\bm{u}\in \ZZ^{\lambda(a)}, \bm{v}\in \ZZ^{\lambda(b)}$ it holds that $\bm{u} + \bm{v} \in \ZZ^{\lambda(ab)}$.
Hence, we define $(a, \bm{u}) \cdot (b, \bm{v}) = 
(ab, \bm{u} + \bm{v})$ for all $(a, \bm{u}), (b, \bm{v})$. The identity element of the resulting monoid $\alg M$ is the element $(e(N), \bm{0})$. It is easy to see that $\alg M$ is commutative and {regular}. The idempotent elements of $\alg M$ are precisely $(d, \bm 0)$ for each $d\in N$, so $\alg M_I$ is isomorphic to $\alg N$, and the $\Hcal$-class of $(d, \bm 0)$ is the subgroup $\{d \} \times \ZZ^{\lambda(d)}$.
\par
Let $\Xi: N \to 2^{\ZZ^Q}$ be such that
$\Xi(d)$ is a subgroup of $\ZZ^{\lambda(d)}$
for each $d\in N$, and $\Xi(a) \subseteq \Xi(b)$ whenever $b \preceq a$.
Then we define $\alg N \times_{\lambda}^{\Xi} \ZZ$ as the quotient
of $\alg N \times_{\lambda} \ZZ$ by the equivalence relation $\cong$ 
defined by $
(a, \bm{u}) \cong (b, \bm{v})$,
whenever $a= b$ and $\bm{u}-\bm{v}\in \Xi(a)$. Given
$(d, \bm{v})\in \alg N \times_{\lambda} \ZZ$, we write $[d, \bm{v}]$
for its corresponding $\cong$-class in $\alg N \times_{\lambda}^{\Xi} \ZZ$. It is not difficult to see that $\alg N \times_{\lambda}^{\Xi} \ZZ$
equipped with the product 
\[
[a, \bm{u}] [b, \bm{v}] = [ab, \bm{u}+\bm{v}] 
\]
is a well-defined monoid. Moreover, $\alg N \times_{\lambda}^{\Xi} \ZZ$
must be commutative and {regular} because it is a homomorphic image of
$\alg N \times_{\lambda} \ZZ$. \par

\begin{restatable}{lemma}{structureabelianregularintegers}
\label{le:structure_abelian_regular_integers}
  Let $\alg M$ be a {regular} commutative monoid generated by a finite set $Q \subseteq M$. Then there are suitable maps  $\lambda: M_I \to 2^Q$ and $\Xi: M_I \to 2^{\ZZ^Q}$ such that
  $\alg M$ is isomorphic to $\alg M_I \times_\lambda^\Xi \ZZ$.
\end{restatable}
\begin{proof}
Let $\lambda: M_I \to 2^Q$ be the map defined by 
    $d \mapsto \{ \alpha \in Q \vert  d\preceq \alpha \}$. 
    We first prove that $\alg M$ is a homomorphic image of $\alg M_I \times_\lambda \ZZ$. Consider the map 
    $\rho: \alg M_I \times_\lambda \ZZ \to \alg M$ given by
    \begin{equation}
    \label{eq:hom_structure_abelian_regular}
           (d, \bm{v}) \mapsto d \prod_{\substack{\alpha \in
           \lambda(d), v_\alpha\neq 0
           }} \alpha^{v_\alpha}.
    \end{equation}
     We claim that $\rho$ is a
    is a monoid homomorphism.
    By construction, $\rho$ preserves the identity element. Let us show that it also preserves products. Another way of writing \eqref{eq:hom_structure_abelian_regular} is
    \[          
    (d, \bm{v}) \mapsto d \prod_{\substack{\alpha \in \lambda(d)}} \alpha^{v_\alpha},
    \]
     where we adopt the convention that $\alpha^0= d_\alpha$.
    To see that this is equivalent to \eqref{eq:hom_structure_abelian_regular},
    observe that, by definition of $\lambda$, the fact that $\alpha\in \lambda(d)$
    implies that $d \preceq d_\alpha$, so $d d_\alpha= d$. 
    Let $(d_1, \bm{u}), (d_2,\bm{v})\in \alg M_I \times_\lambda \ZZ$. Then
    \begin{align*}
    \rho(d_1, \bm{u}) \rho(d_2,\bm{v}) = 
    d_1 d_2 \prod_{\substack{\alpha \in \lambda(d_1) \cup \lambda(d_2)}} \alpha^{u_\alpha+ v_\alpha} \\
    =  d_1 d_2 \prod_{\substack{\alpha \in \lambda(d_1d_2)}} \alpha^{(u+v)_\alpha} =  \rho(d_1 d_2, \bm{u}+\bm{v}).
    \end{align*}
    This shows that $\rho$ is a monoid homomorphism. Now, it is also not hard to see that $\rho$ is surjective. Let $a\in \alg M$. Then $a$ can be expressed as $\prod_{\alpha\in Q_a} \alpha^{n_\alpha}$ for some 
    subset of generators $Q_a\subseteq Q$ and some integers $n_\alpha>0$
    for each $\alpha\in Q_a$. In particular, it must be that $d_a
    \preceq \alpha$ for each $\alpha \in Q_a$, so $Q_a\subseteq \lambda(d_a)$. Additionally, the fact that $a= d_a a$ implies
    $a= d_a\prod_{\alpha\in Q_a} \alpha^{n_\alpha}$. Define $\bm{v}\in \ZZ^{\lambda(d_a)}$ by letting $v_\alpha= n_\alpha$ if $\alpha\in Q_a$
    and $v_\alpha = 0$ otherwise. Then $a= \rho(d_a, \bm{v})$.\par
    The fact that $\alg M= \im(\rho)$ implies that $\alg M$ is isomorphic to
    the quotient of $\alg M_I \times_\lambda \ZZ$ by the congruence $\cong$
    given by $(a,\bm u)\cong (b, \bm v)$ whenever $\rho(a, \bm u)=\rho(b,\bm v)$. Let us analyze this congruence. The first observation is that if $c = \rho(d, \bm v)$, then $d_c=d$. This shows that $(a,\bm u)\cong (b, \bm v)$ implies that necessarily $a=b$. For each $d\in M_I$, define $\Xi(d)\subseteq \ZZ^{\lambda(d)}$ 
    as the submodule consisting of the vectors $\bm v$ satisfying
    $\rho(d, \bm v)= \rho(d, \bm 0)$. By construction
    $(a,\bm u)\cong (b, \bm v)$ if and only if $a=b$, and 
    $\bm u- \bm v \in \Xi(a)$. This shows that $\alg M$ is isomorphic
    to $\alg M_I \times_{\lambda}^\Xi \ZZ$.
\end{proof}

In the following result we will consider the expansion of a monoid $\alg M$ of the form $\alg N \times_\lambda^\Xi \ZZ$ by a relation $R\subseteq M^r$. 
It will be convenient to represent tuples $([d_1, \bm v_r],\dots,
[d_r, \bm v_r])\in M^r$ as pairs $[\bm d, \bm u]$ where $\bm d= (d_1,\dots,
d_r)$ and $\bm u\in \prod_{i\in [r]} \ZZ^{\lambda(d_i})$ is a vector
whose projection to $\ZZ^{\lambda(d_i)}$ 
is $\bm v_i$ for each $i\in [r]$. This way, we identify $M^r$ with the set
\[
\bigsqcup_{\bm d \in N^r} \{\bm d \} \times \prod_{i\in [r]} \ZZ^{\lambda(d_i)}/ \Xi(d_i).
\]

\begin{restatable}{theorem}{algorithm}\label{thm:algorithm}
    If $\rel M$ is an expansion of a finitely generated {regular} commutative monoid $\alg M$ by a coset $R\subseteq M^r$, then $\csp(\rel M)$ is solvable in polynomial time.
\end{restatable}

\begin{proof}
 Let $\rel X$ be an instance of $\csp(\rel M)$. We can see $\rel X$ as a set of variables $X$ together with expressions of the form (1) $x_1 x_2 = x_3$ for $x_1,x_2, x_3\in X$, (2) $x=e$ for $x\in X$, and (2) $(x_1,\dots, x_r)\in R$ for $x_1,\dots, x_r\in X$.  \par

By \Cref{structure-abelian-monoid}, the fact that $\alg M$ is finitely generated implies that the submonoid $\alg M_I \leq \alg M$
of idempotent elements is finite. Hence, by
\Cref{le:structure_abelian_regular_integers}, we can assume that
$\alg M$ is of the form
$\alg N \times_{\lambda}^{\Xi} \ZZ$, where $\alg N$ is a finite semilattice, $\lambda: N \to 2^Q$ for some finite set $Q$, and
$\Xi: N \to 2^{\ZZ^Q}$. Recall that $\alg M= \bigsqcup_{\delta\in M_I} \alg H_\delta$, where $H_\delta$ denotes the $\Hcal$-class of the idempotent $M_I$, which is a subgroup of $\alg M$. As
$\alg M = \alg N \times_{\lambda}^{\Xi} \ZZ$, we have that 
$M_I = \{  [\delta, \bm{0}]\, \vert \, \delta \in  N \}$. We abuse the notation and identify $\alg N$ with $\alg M_I$ via the isomorphism $\delta \mapsto [\delta,\bm 0]$.
This way, given $\delta\in M_I$, the subgroup $H_\delta$ is defined as
$\{ [\delta, \bm v] \, \vert \, \bm v \in \ZZ^{\lambda(\delta)} 
\}$.
\par
Given $a\in M$, we write $\delta_a$ for the unique idempotent element satisfying $\delta_a\Hcal a$. Similarly, if $\bm a \in M^r$, we write $\bm{\delta}_{\bm a}$ for the unique idempotent element in the $\Hcal$-class of $\bm a$ (such element exists because $M^r$ is a regular commutative monoid). Let $\rel M_I$ be the expansion of $\alg M_I$ by the relation $R(M_I) = \{ \bm{\delta}_{\bm a} \vert \bm a \in R \}$. Then the projection $\pi_I: \rel M\to \rel M_I$ defined by $a \mapsto \delta_a$ is a homomorphism.  
    Suppose that $f: \rel X \to \rel M$ is a homomorphism. Then $\pi_I\circ f$
    is a homomorphism from $\rel X$ to $\rel M_I$. Additionally, if $h: \rel X
    \to \rel M_I$ is a homomorphism, then the product $fh$ is also a
    homomorphism from $\rel X$ to $\rel M$. Indeed, $fh$ clearly preserves products and sends $e(X)$ to $e(M)$. To see that $fh$ preserves the relation $R$, observe that if $R(X)$ contains a tuple $\bm{x}$, then $f(\bm{x})\in R(M)$ and $h(\bm{x})\in R(M_I) \subseteq (R(M))^{-1} \otimes R(M)$. This way, using the fact that $R(M)$ is a coset, we obtain 
\[ fh(\bm{x})= f(\bm{x})h(\bm{x}) \in R(M) \otimes ((R(M))^{-1} \otimes R(M)) \subseteq R(M).
\]
Additionally, this new solution $fh$ satisfies 
$\pi_I \circ (fh)(x) \preceq h(x)$ for all $x\in X$. 
Now, observe that if $\rel X \to \rel M_I$ then there must be a 
\emph{minimal} homomorphism $h:\rel X \to \rel M_I$, meaning that
for any $h^\prime: \rel X \to \rel M_I$ it holds that
$h(x)\preceq h^\prime(x)$ for all $x\in X$. To see this define
$h(x) = \prod_{h^\prime: \rel X \to \rel M_I} h^\prime(x)$.
This way, if there is a solution $f^\prime: \rel X \to \rel M$ then, by defining $f=f^\prime h$, we can obtain a solution $f: \rel X \to \rel M$ such that $\pi_I\circ f$ is the minimal homomorphism $h: \rel X  \to \csp(\rel M_I)$. Our algorithm aims to find such a solution $f$.\par
  By {\cite[Lemma~19]{KTT07:tcs}}, we can find a minimal solution $h: \rel X \to \rel M_I$ in polynomial time if one exists. Otherwise we reject the instance $\rel X$. 
Now we look for a map $f: \rel X \to \rel M$ 
such that $f(x)$ is of the form $[\delta, \bm{v}]\in H_\delta$ for some, whenever $h(x)= \delta$. 
Next we construct a system $\Sigma_X$ of linear equations over $\ZZ$. For each variable $x\in X$, we include a vector $\bm{v}^x$ of integer variables  $(v^x_\alpha)_{\alpha\in \lambda(h(x))}$. The intuition is that our intended solution $f$ should be defined as $f(x)=[h(x), \bm{v}^x]$.
For each $x\in X$ and each $\alpha \notin \lambda(h(x))$ we add to $\Sigma_X$ an equation $v^x_\alpha=0$. 
For each constraint in $\rel X$ of the form
$x=e$, we add to $\Sigma_X$ equations ensuring $\bm{v}^x=\bm{0}$.

We make use of the well-known result that any subgroup of a finitely generated
  commutative group is itself finitely generated \cite[Theorem 1.6, Chapter
  II]{hungerford2012algebra}.
For each $\delta\in M_I$, let $W_\delta$ be a finite set generating the subgroup $\Xi(\delta) \subseteq \ZZ^{\lambda(\delta)}$. 
For each constraint in $\rel X$ of the form $xy=z$, we add to $\Sigma_X$ equations
ensuring 
\[\bm{v}^x + \bm{v}^y = \bm{v}^z + \sum_{\bm{u}\in W_{h(z)}} n_{\bm u} \bm{u}\]
where
$n_{\bm u}$ is a fresh variable for each $\bm{u}\in W_{h(z)}$. This simply ensures that 
\[
[h(x), \bm{v}^x] + [h(y), \bm{v}^y] = [h(z), \bm{v}^z].
\]
For each $\bm{\delta}\in \pi_I(R)$,
let $R_{\bm{\delta}}$ be $R \bigcap H_{\bm{\delta}}$. Using the fact that $R$ is a coset of $\alg M^r$, we obtain that $R_{\bm{\delta}}$ must be a coset of $H_{\bm{\delta}}$. This way, there must be a coset $U_{\bm{\delta}}\subseteq \prod_{i\in [r]} \ZZ^{\lambda(\delta_r)}$ such that $R_{\bm{\delta}}= \{ \bm \delta\} \times U_{\bm{\delta}}$. It is not hard to show that a coset $U$ in a commutative group $\alg G$ must be of the form $a + H$ (using additive notation), where $a\in G$ and $\alg H\leq \alg G$ is a subgroup. Using again the fact that subgroups of finitely generated commutative groups must be finitely generated, we conclude there is a vector $\bm{o}_{\bm \delta} \in \prod_{i=1}^r \ZZ^{\lambda(\delta_i)}$
and $V_{\bm \delta} \subseteq \prod_{i=1}^r \ZZ^{\lambda(\delta_i)}$ a finite subset such that
\[
R_{\bm{\delta}}= \{ 
[\bm{\delta}, \bm{u}] \, \vert \, 
\bm{u} \in \bm{o}_{\bm \delta} + \langle  V_{\bm \delta} \rangle
\}.
\]
Then, for each constraint in $\rel X$ of the form
$\bm{x}=(x_1,\dots, x_r) \in R$, we add to $\Sigma_X$ equations ensuring
\[
\bm{v}^{\bm{x}} = 
\bm{o}_{h(\bm x)} + \sum_{\bm{u}\in V_{h(\bm x)}} n_{\bm u} \bm u,
\]
where $\bm{v}^{\bm{x}}$ denotes the vector $(\bm v^{x_1},\dots, \bm v^{x_r})$, and $n_{\bm u}$ is a fresh variable for each $\bm{u}\in V_{h(\bm x)}$. This ensures that $[h(\bm x), \bm{v}^{\bm x}]\in R(M)$. Now, by construction, $\Sigma_X$ has a solution over the integers if and only if $\rel X$ has a solution in $\csp(\rel M)$. Systems of linear equations over the integers can be solved in polynomial time via Gaussian elimination, so this completes the proof. 
\end{proof}

\appendix

\section{Auxiliary Results About Monoids}
\label{ap:algebra}

\structureabelianmonoid*
\begin{proof}
Item (1) is straightforward. Let us prove (2). 
If $\alg M$ is {regular}, then the quotient $\alg M/\Hcal$
is isomorphic to $\alg M_I$ through the map that sends each $\Hcal$-class
to the single idempotent element belonging to it. Hence if $\alg M$ is finitely generated, so is $\alg M_I$. Since $\alg M_I$ is a finitely generated semilattice, it is finite. 
\end{proof}

The following is a characterization of group elements in finite {regular} commutative monoids. 

\begin{lemma}[{\cite[Lemma 1]{LZ24:acm}}]
\label{le:ab_reg_characterization}
Let $\alg M$ be a finite commutative monoid, and $a\in \alg M$ an element. Then the following are equivalent.
\begin{itemize}
    \item[(1)] $a^2b=a$ for some $b\in \alg M$,
    \item[(2)] $a^k=a$ for some $k>1$,
    \item[(3)] $a$ is a group element,
    \item[(4)] $a\preceq a^2$.
\end{itemize}
\end{lemma}

\idempotentcst*
\begin{proof}
    Let $a\in \alg M$. As $\alg M$ is finite, there must be some $\ell_a, k_a \in \NN$ satisfying that
    \begin{equation}
    \label{eq:idempotent_constant_aux}
    a^{\ell_a + n} = a^{\ell_a + m},
    \end{equation}
    for any integers $n,m\geq 0$ satisfying $m = n \mod k_a$ (see e.g., \cite{clifford1961algebraic}[Section 1.6]). 
    This way, the element
    $a^{k_a\ell_a}$ is idempotent. Indeed, 
    \[
    a^{2k_a\ell_a} =  a^{\ell_a + (2k_a-1)\ell_a} = a^{\ell_a + (k_a-1)\ell_a} = a^{k_a\ell_a}.
    \]
    Here, the second equality uses \eqref{eq:idempotent_constant_aux}. The constant $C$ that witnesses the result can be defined as $C= \prod_{a\in \alg M} k_a\ell_a$.
    If $a\in M$ is a group element, then it belongs to a subgroup 
    whose identity must be $a^C$, so $a^{-1}= a^{C-1}$. \par
    Fix $a\in M$. We show that there is a unique idempotent element $d_a$ of the form
    $a^n$ for some $n\in \NN$. We have shown that there is at least one idempotent $a^n$ of this form. Suppose that there is another positive integer $m\neq n$ such that
    $a^m$ is idempotent as well. It holds that $a^m\Hcal a^n$. Indeed, there are numbers $k, \ell\in \NN$ satisfying $km> n$ and $\ell n > m$, so 
    \[
    a^n = a^{\ell n} = a^m a^{\ell n - m} = a^{\ell n - m} a^m,
    \]
    and
    \[
    a^m = a^{k m} = a^n a^{k m - n} = a^{k m - n} a^n.
    \]
    By \Cref{th:green_theorem}, there is at most a single idempotent element in each $\Hcal$-class of $\alg M$, so $a^m= a^n$. This completes the proof. 
\end{proof}

\regsubmonoid*
\begin{proof}
\textbf{(1)} The map $\pi_I$ preserves the identity element, because $d_{e_M}= e_M$. To see that $\pi_I$ preserves the product, let $k\in \NN$ be such that $a^k= d_a$ for each $a\in \alg M$, as given by \Cref{le:idempotent_constant}. Then $d_a d_b = a^k b^k = (ab)^k = d_{ab}$ for each $a, b\in \alg M$. 
\textbf{(2)} Let $k\in \NN$ be such that $a^k$ is idempotent for all $a\in \alg M$, as given by \Cref{le:idempotent_constant}. We have that $a\in \alg M$ is a group element if and only if $a^{k+1}=a$. Indeed, if $a^{k+1}=a$ then $a$ is a group element by \Cref{le:ab_reg_characterization}. Conversely, if $a$ is a group element, by \Cref{th:green_theorem} its $\Hcal$-class is a subgroup of $\alg M$. The identity of this subgroup must be $a^k$, so $a^{k+1}=a$. \textbf{(3)} The map $\pi_\dagger$ clearly preserves the identity element. To see that it also preserves products, observe that by (1), $d_ad_b=d_{ab}$ for every $a,b\in \alg{M}$, so $a d_a b d_b= ab d_{ab}$. Finally, by (1), it holds that
    $d_{ad_a} = d_a d_{d_a} = d_a^2 = d_a$ for all $a\in \alg M$. This way, $\pi_\dagger \circ \pi_\dagger(a)= (ad_a) d_{ad_a} = a d_a d_a  = ad_a$ for all $a\in \alg M$.
\end{proof}

\splitaux*
\begin{proof}
    Let $a\in \Cos{g(S)_\dagger}$. Then we can express $a$ as a product of the form $g(b)_\dagger \prod_{i\in [k]} g(s_i)_\dagger g(t_i)_\dagger^{-1}$, for some elements $b, s_i, t_i\in S$.
    By \Cref{le:regular_submonoid}, for each $s\in \alg M$ it holds that
    $g(s)_\dagger = \prod_{f\in F} f(s)_\dagger$. Hence, we can write
    \[
    a = \prod_{f\in F} \left( f(b)_\dagger \prod_{i\in [k]} f(s_i)_\dagger f(t_i)_\dagger^{-1}\right). 
    \]
    This shows that $a\in  \bigotimes_{f\in F} \Cos{f(S)_\dagger}$, as we wanted to prove. 
\end{proof}

\section{Commutative Regularization}
\label{ap:ab_reg}

The aim of this section is to prove \Cref{le:abreg_definition}. More explicitly, given a monoid $\alg M$, we construct its \emph{commutative regularization}.
This is a second monoid $\abreg{\alg M}$ (r.c. standing for \textit{regular commutative}) together with a homomorphism $f\colon \alg M \rightarrow \abreg{\alg M}$ satisfying the following universal property: 
for every {regular} commutative monoid $\alg N$ and every homomorphism $g\colon \alg M\rightarrow \alg N$, there exists a unique homomorphism $h\colon \abreg{\alg M} \rightarrow \alg N$ satisfying $g = f\circ h$. \par
We use the notion of a monoid presented by a set of generators and relations~\cite{howie1995fundamentals}. This way, in a monoid $\alg M$
presented by a set of generators $S$ and some relations, the elements of $\alg M$ are equivalence classes of $S^*$, where $S^*$ denotes the set of non-empty words over $S$. Given a word $\alpha\in S^*$, we write 
$[\alpha]$ to denote its equivalence class in $\alg M$. \par

Let $\alg M$ be a monoid.
Informally, to define $\abreg{\alg M}$ we add inverses to each element $a\in \alg M$ and we impose relations that make the resulting monoid commutative.  We define $\abreg{\alg M}$ as the monoid presented by the set of generators $\Omega$
that contains the symbols
$\widehat{a},\widecheck{a}, 1_a$ for each element $a\in \alg M$, 
and the set of relations containing all identities of the form (1)
$
\prod_{i\in [k]} \widehat{a}_i = \prod_{j\in [\ell]} \widehat{b}_j $
for each equality
$\prod_{i\in [k]} a_i = 
\prod_{j\in [\ell]} b_\ell$
that holds in $\alg M$, 
(2) $\widehat{a} 1_a = \widehat{a}$
for each $a\in M$,
(3) $\widehat{a} \widecheck{a} = 1_a$
for each $a\in M$,
(4) $\widehat{e_{\little{M}}} = e$,
and
(5) $xy= yx$ for each $x,y\in \Omega$.

\par

\begin{lemma}
\label{le:abreg_definition} 
The monoid $\abreg{ \alg M}$
is commutative and {regular}. Moreover, for each $a\in \alg M$, the idempotent element in the $\Hcal$-class of $[\widehat{a}]$ is
$[1_a]$, and $[\widehat{a}]^{-1}= [\widecheck{a}]$.    
\end{lemma}
\begin{proof}
Commutativity follows from (5). In order to see that $\abreg{ \alg M}$ is {regular}, first observe that the elements $[1_a]$ are idempotent for all $a\in M$: 
\[
[1_a^2]= [1_a \widehat{a}\widecheck{a}]=
[\widehat{a}\widecheck{a}]=
[1_a].
\]
Now let us see that $\abreg{\alg M}$ is {regular}.
By \Cref{prop:charRM} we just need to show that there is an idempotent element in each $\Hcal$-class of $\abreg{ \alg M}$. As $\abreg{ \alg M}$ is a commutative monoid and $\{[a] \vert a\in \Omega \}$ is a set of generators, any non-identity element $c$ can be expressed as $\prod_{i\in [k]} [b_i]^{n_i}$ for some elements $b_1,\dots, b_k\in \Omega$ and some integers $n_1,\dots, n_k\geq 1$. For each $i\in [k]$
define $b_i^\prime= \widecheck{a}$ if $b_i= \widehat{a}$, $b_i^\prime = \widehat{a}$ if
$b_i= \widecheck{a}$, or $b_i^\prime = b_i$ if $b_i = 1_a$. This way $[b_i] [b_i^\prime]$ is an idempotent element, and
$[b_i]^2[b_i^\prime]=[b_i]$. We define $c^\prime =  \prod_{i\in [k]} [b_i^\prime]^{n_i}$. It holds that $c c^\prime$ is idempotent, and
$c^2 c^\prime= c$, proving that $\abreg{\alg M}$ is {regular}. \par
To prove the last part of the statement, observe that both 
$[\widehat{a}][\widecheck{a}]= [1_a]$ and
$[\widehat{a}]^2[\widecheck{a}]= [\widehat{a}]$ hold for all $a\in \alg{M}$, and $[1_a]$ is an idempotent element.
\end{proof}

\begin{lemma}
\label{le:generators_abreg}
  Let $S \subseteq M$ be a set of generators of a monoid $\alg M$. Then the set 
\[
\widetilde{S}= 
\{ [\widehat{a}] \mid a\in S \} \cup \{ 
[\widecheck{a}] \mid a\in S\}
\]
is a set of generators of $\abreg{\alg M}$.
\end{lemma}
\begin{proof}
We show that every element of the form $[b]$ for $b\in \Omega$ can be generated by these elements.
Let $c \in M$ be an arbitrary element. By assumption $c=t^{\alg M}(a_1,\dots, a_k)$ for some term $t$ and some elements $a_1, \dots, a_k\in S$. Suppose that $b= \widehat{c}$. 
Then, because $\abreg{ \alg M}$ satisfies all identities in (1), it must hold that $[b]=t^{\abreg{\alg M}}([\widehat{a}_1], \dots, [\widehat{a}_k])$.
Suppose that $b=\widecheck{c}$. We know that $[\widecheck{a}]=[\widehat{a}]^{-1}$ for all $a\in \alg M$, and that the map that sends each element to its inverse is an endomorphism of $\abreg{\alg M}$. Hence
it must be that
$[b]=t^{\abreg{\alg M}}([\widecheck{a}_1], \dots, [\widecheck{a}_k])$.
Finally, suppose that $b=1_c$. In other words, $[b]$ is the idempotent element
  in $H_{[\widehat{c}]}$. Recall that the map sending each element $a\in \abreg{\alg M}$ to the idempotent element in $H_a$ is a homomorphism. Hence,
  $t^{\abreg{\alg M}}([1_{a_1}], \dots, [1_{a_k}])$ must be the idempotent element in $H_{[\widehat{c}]}$, which equals $[b]$. Observe that $[1_{a}]=[\widehat{a}][\widecheck{a}]$ for each $a\in M$, so each $[1_{a_i}]$ is generated by the elements in $\widetilde{S}$. This completes the proof. 
\end{proof}

Consider the map $\abreg{\pi_M}: \alg M \rightarrow \abreg{\alg M}$
defined by $a\mapsto [\widehat{a}]$.
By definition $\abreg{\alg M}$ satisfies all identities in (1), so $\abreg{\pi_M}$ preserves products. By the identity (4), $[\widehat{e_M}]$ is the identity element in $\abreg{M}$, so $\abreg{\pi_M}$ is a monoid homomorphism.

\begin{lemma}\label{le:univ-prop-abreg}
    Let $f: \alg M \rightarrow \alg N$
    be a monoid homomorphism, where $\alg N$ is a regular commutative monoid. Then there is another monoid homomorphism $g: \abreg{\alg M} \rightarrow \alg N$
    such that $g\circ \abreg{\pi_M}=f$.
\end{lemma}
\begin{proof}
    Let $\Omega$ be the set of generators in the definition of $\abreg{\alg M}$. We define a map $h: \Omega \rightarrow \alg N$ as follows.  For each $a\in M$ let $h(\widehat{a})= f(a)$,  
    $h(\widecheck{a})= f(a)^{-1}$, and
    $h(1_a)= d_{f(a)}$. It is easy to see that for each identity
    $t(x_1,\dots, x_k)= s(y_1,\dots, y_\ell)$ in the presentation of $\abreg{\alg M}$, the identity
    $t^{\alg N}(h(x_1),\dots, h(x_k)) = s^{\alg N}(h(y_1),\dots, h(y_\ell))$ holds. This way, by the universal property of a semigroup given by a presentation, the
    map $g: \abreg{\alg M} \rightarrow \alg N$ given by $[x] \mapsto h(x)$ is a well-defined semigroup homomorphism from $\abreg{\alg M}$ to $\rel N$. To see that it is also a monoid homomorphism, observe that $h(\widehat{e})= f(e)$, and $f(e)$ must be the identity element of $\alg N$. \par
    Finally, the fact that $g\circ \abreg{\pi_M}= f$ is straightforward. Given
    $s\in M$, we have $\abreg{\pi_M}(s)= [\widehat{s}]$, and $g([\widehat{s}])= h(\widehat{s})=f(s)$, as we wanted to show. This completes the proof. 
\end{proof}

{\small
\bibliographystyle{plainurl}
\bibliography{lmz}

\begin{thebibliography}{10}

\bibitem{AB21}
Kristina Asimi and Libor Barto.
\newblock Finitely tractable promise constraint satisfaction problems.
\newblock In {\em Proc. 46th International Symposium on Mathematical
  Foundations of Computer Science (MFCS'21)}, volume 202 of {\em LIPIcs}, pages
  11:1--11:16. Schloss Dagstuhl -- Leibniz-Zentrum f{\"{u}}r Informatik, 2021.
\newblock \href {http://arxiv.org/abs/2010.04618} {\path{arXiv:2010.04618}},
  \href {https://doi.org/10.4230/LIPIcs.MFCS.2021.11}
  {\path{doi:10.4230/LIPIcs.MFCS.2021.11}}.

\bibitem{UnificationACUI}
Franz Baader, Pavlos Marantidis, Antoine Mottet, and Alexander Okhotin.
\newblock Extensions of unification modulo {ACUI}.
\newblock {\em Math. Struct. Comput. Sci.}, 30(6):597--626, 2020.
\newblock \href {https://doi.org/10.1017/S0960129519000185}
  {\path{doi:10.1017/S0960129519000185}}.

\bibitem{HandbookAutomatedReasoningBaader}
Franz Baader and Wayne Snyder.
\newblock Unification theory.
\newblock In {\em Handbook of Automated Reasoning}, pages 447--533. Elsevier,
  2001.

\bibitem{BBKO21}
Libor Barto, Jakub Bul{\'\i}n, Andrei~A. Krokhin, and Jakub Opr{\v s}al.
\newblock Algebraic approach to promise constraint satisfaction.
\newblock {\em J. {ACM}}, 68(4):28:1--28:66, 2021.
\newblock \href {http://arxiv.org/abs/1811.00970} {\path{arXiv:1811.00970}},
  \href {https://doi.org/10.1145/3457606} {\path{doi:10.1145/3457606}}.

\bibitem{BartoDemeoMottet}
Libor Barto, William~J. DeMeo, and Antoine Mottet.
\newblock Constraint satisfaction problems over finite structures.
\newblock In {\em Proc. 36th Annual {ACM/IEEE} Symposium on Logic in Computer
  Science (LICS'21)}, pages 1--13. {IEEE}, 2021.
\newblock \href {https://doi.org/10.1109/LICS52264.2021.9470670}
  {\path{doi:10.1109/LICS52264.2021.9470670}}.

\bibitem{Barto22:soda}
Libor Barto and Marcin Kozik.
\newblock {Combinatorial Gap Theorem and Reductions between Promise {CSP}s}.
\newblock In {\em Proc. 2022 ACM-SIAM Symposium on Discrete Algorithms
  (SODA'22)}, pages 1204--1220, 2022.
\newblock \href {http://arxiv.org/abs/2107.09423} {\path{arXiv:2107.09423}},
  \href {https://doi.org/10.1137/1.9781611977073.50}
  {\path{doi:10.1137/1.9781611977073.50}}.

\bibitem{FiniteAlgebras}
Libor Barto and Antoine Mottet.
\newblock Finite algebras with hom-sets of polynomial size.
\newblock {\em Trans. Amer. Math. Soc.}, 378:569--596, 2025.
\newblock \href {https://doi.org/10.1090/tran/9262}
  {\path{doi:10.1090/tran/9262}}.

\bibitem{CSPExpansionsFreeAbelianGroups}
Manuel Bodirsky, Barnaby Martin, Marcello Mamino, and Antoine Mottet.
\newblock The complexity of disjunctive linear {D}iophantine constraints.
\newblock In {\em Proc. 43rd International Symposium on Mathematical
  Foundations of Computer Science (MFCS'18)}, volume 117 of {\em LIPIcs}, pages
  33:1--33:16. Schloss Dagstuhl - Leibniz-Zentrum f{\"{u}}r Informatik, 2018.
\newblock \href {https://doi.org/10.4230/LIPICS.MFCS.2018.33}
  {\path{doi:10.4230/LIPICS.MFCS.2018.33}}.

\bibitem{BG21:sicomp}
Joshua Brakensiek and Venkatesan Guruswami.
\newblock {Promise Constraint Satisfaction: Algebraic Structure and a Symmetric
  Boolean Dichotomy}.
\newblock {\em {SIAM} J. Comput.}, 50(6):1663--1700, 2021.
\newblock \href {http://arxiv.org/abs/1704.01937} {\path{arXiv:1704.01937}},
  \href {https://doi.org/10.1137/19M128212X} {\path{doi:10.1137/19M128212X}}.

\bibitem{BGWZ20}
Joshua Brakensiek, Venkatesan Guruswami, Marcin Wrochna, and Stanislav
  {\v{Z}}ivn{\'y}.
\newblock The power of the combined basic {LP} and affine relaxation for
  promise {CSP}s.
\newblock {\em {SIAM} J. Comput.}, 49:1232--1248, 2020.
\newblock \href {http://arxiv.org/abs/1907.04383} {\path{arXiv:1907.04383}},
  \href {https://doi.org/10.1137/20M1312745} {\path{doi:10.1137/20M1312745}}.

\bibitem{BLZ25:icalp}
Silvia Butti, Alberto Larrauri, and Stanislav \v{Z}ivn{\'{y}}.
\newblock Optimal inapproximability of promise equations over finite groups.
\newblock In {\em Proc. 52nd International Colloquium on Automata, Languages,
  and Programming (ICALP'25)}. Schloss Dagstuhl -- Leibniz-Zentrum f{\"u}r
  Informatik, 2025.
\newblock \href {http://arxiv.org/abs/2411.01630} {\path{arXiv:2411.01630}}.

\bibitem{clifford1961algebraic}
Alfred~H Clifford and Gordon~B Preston.
\newblock The algebraic theory of semigroups, vol. 1.
\newblock {\em Mathematical surveys}, 7, 1961.

\bibitem{ExistentialTheoryFreeGroups}
Volker Diekert, Claudio Gutierrez, and Christian Hagenah.
\newblock {The existential theory of equations with rational constraints in
  free groups is PSPACE-complete}.
\newblock {\em Inf. Comput.}, 202(2):105--140, 2005.
\newblock \href {https://doi.org/10.1016/J.IC.2005.04.002}
  {\path{doi:10.1016/J.IC.2005.04.002}}.

\bibitem{Engebretsen04:tcs}
Lars Engebretsen, Jonas Holmerin, and Alexander Russell.
\newblock Inapproximability results for equations over finite groups.
\newblock {\em Theor. Comput. Sci.}, 312(1):17--45, 2004.
\newblock \href {https://doi.org/10.1016/S0304-3975(03)00401-8}
  {\path{doi:10.1016/S0304-3975(03)00401-8}}.

\bibitem{Feder98:monotone}
Tom\'as Feder and Moshe~Y. Vardi.
\newblock The computational structure of monotone monadic {S{N}{P}} and
  constraint satisfaction: {A} study through {D}atalog and group theory.
\newblock {\em {SIAM} J. Comput.}, 28(1):57--104, 1998.
\newblock \href {https://doi.org/10.1137/S0097539794266766}
  {\path{doi:10.1137/S0097539794266766}}.

\bibitem{Gillibert22:sicomp}
Pierre Gillibert, Julius Jonusas, Michael Kompatscher, Antoine Mottet, and
  Michael Pinsker.
\newblock When symmetries are not enough: {A} hierarchy of hard constraint
  satisfaction problems.
\newblock {\em {SIAM} J. Comput.}, 51(2):175--213, 2022.
\newblock \href {http://arxiv.org/abs/2002.07054} {\path{arXiv:2002.07054}},
  \href {https://doi.org/10.1137/20M1383471} {\path{doi:10.1137/20M1383471}}.

\bibitem{SkolemArithmetic}
Christian Gla{\ss}er, Peter Jonsson, and Barnaby Martin.
\newblock Circuit satisfiability and constraint satisfaction around {S}kolem
  arithmetic.
\newblock {\em Theor. Comput. Sci.}, 703:18--36, 2017.
\newblock \href {https://doi.org/10.1016/j.tcs.2017.08.025}
  {\path{doi:10.1016/j.tcs.2017.08.025}}.

\bibitem{Goldmann02:ic}
Mikael Goldmann and Alexander Russell.
\newblock The complexity of solving equations over finite groups.
\newblock {\em Inf. Comput.}, 178(1):253--262, 2002.
\newblock \href {https://doi.org/10.1006/INCO.2002.3173}
  {\path{doi:10.1006/INCO.2002.3173}}.

\bibitem{grillet2017semigroups}
Pierre~A Grillet.
\newblock {\em Semigroups: an introduction to the structure theory}.
\newblock Routledge, 2017.

\bibitem{Hastad01}
Johan H{\aa}stad.
\newblock Some optimal inapproximability results.
\newblock {\em J. {ACM}}, 48(4):798--859, 2001.
\newblock \href {https://doi.org/10.1145/502090.502098}
  {\path{doi:10.1145/502090.502098}}.

\bibitem{howie1995fundamentals}
John~M Howie.
\newblock {\em Fundamentals of semigroup theory}.
\newblock Oxford University Press, 1995.

\bibitem{hungerford2012algebra}
Thomas~W Hungerford.
\newblock {\em Algebra}, volume~73.
\newblock Springer Science \& Business Media, 2012.

\bibitem{Kannan}
Ravindran Kannan and Achim Bachem.
\newblock Polynomial algorithms for computing the {S}mith and {H}ermite normal
  forms of an integer matrix.
\newblock {\em {SIAM} J. Comput.}, 8(4):499--507, 1979.
\newblock \href {https://doi.org/10.1137/0208040} {\path{doi:10.1137/0208040}}.

\bibitem{KTT07:tcs}
Ond\v{r}ej Kl{\'{\i}}ma, Pascal Tesson, and Denis Th{\'{e}}rien.
\newblock Dichotomies in the complexity of solving systems of equations over
  finite semigroups.
\newblock {\em Theory Comput. Syst.}, 40(3):263--297, 2007.
\newblock \href {https://doi.org/10.1007/S00224-005-1279-2}
  {\path{doi:10.1007/S00224-005-1279-2}}.

\bibitem{KO22:survey}
Andrei Krokhin and Jakub Opr{\v{s}}al.
\newblock An invitation to the promise constraint satisfaction problem.
\newblock {\em ACM SIGLOG News}, 9(3):30--59, 2022.
\newblock \href {http://arxiv.org/abs/2208.13538} {\path{arXiv:2208.13538}}.

\bibitem{LZ24:acm}
Alberto Larrauri and Stanislav {\v{Z}}ivn{\'{y}}.
\newblock Solving promise equations over monoids and groups.
\newblock {\em ACM Trans. Comput. Log.}, 26(1):1--24, 2024.
\newblock \href {http://arxiv.org/abs/2402.08434} {\path{arXiv:2402.08434}},
  \href {https://doi.org/10.1145/3698106} {\path{doi:10.1145/3698106}}.

\bibitem{Mottet25:lics}
Antoine Mottet.
\newblock Algebraic and algorithmic synergies between promise and
  infinite-domain {C}{S}{P}s.
\newblock In {\em Proc. 40th Annual {ACM/IEEE} Symposium on Logic in Computer
  Science (LICS'25)}. {IEEE}, 2025.
\newblock \href {http://arxiv.org/abs/2501.13740} {\path{arXiv:2501.13740}}.

\bibitem{Plandowski04:jacm}
Wojciech Plandowski.
\newblock Satisfiability of word equations with constants is in {PSPACE}.
\newblock {\em J. {ACM}}, 51(3):483--496, 2004.
\newblock \href {https://doi.org/10.1145/990308.990312}
  {\path{doi:10.1145/990308.990312}}.

\end{thebibliography}
}

\end{document}